\documentclass[a4paper,reqno]{amsart}

\usepackage[all]{xy}           %For commutative diagrams
\usepackage{amssymb}           %For double arrows
\usepackage{hyperref}
\usepackage{eucal}
\usepackage{graphicx}
\usepackage{epsfig}
\usepackage[usenames,dvipsnames]{color}
\usepackage{color}
 \usepackage{subfigure}

\numberwithin{equation}{section}

\newtheorem{definition}{Definition}[section]
\newtheorem{lemma}[definition]{Lemma}
\newtheorem{theorem}[definition]{Theorem}
\newtheorem{proposition}[definition]{Proposition}

\newtheorem{remarkth}[definition]{Remark}

\newenvironment{remark}{\begin{remarkth}\upshape}{\hfill$\diamond$\end{remarkth}}
\renewcommand{\emph}[1]{{\bfseries\itshape{#1}}}
\newcommand{\R}{\mathbb{R}}      %Numeros reales
\newcommand{\Z}{\mathbb{Z}}      %Numeros enteros
\newcount\ancho \newcount\anchom \newcount\anchoa
\newcount\anchob \newcount\altura

\newcommand{\ltilde}[3][0]{\altura=0 \advance\altura by #1
           \ancho=#2 \anchom=\ancho \divide\anchom by 2
           \anchoa=\ancho \divide\anchoa by 4
           \anchob=\anchom \advance\anchob by \anchoa
           \kern-3pt \begin{array}[b]{c}
           \begin{picture}(1,1)(\anchom,-\altura)
        \qbezier(0,2)(\anchoa,5)(\anchom,2)
        \qbezier(\anchom,2)(\anchob,-1)(\ancho,4)
        \qbezier(0,2)(\anchoa,4.5)(\anchom,1.8)
        \qbezier(\anchom,1.8)(\anchob,-1.5)(\ancho,4)
       \end{picture} \\[-4pt]{#3}
                       \end{array} \kern-4pt    }

\newcommand{\lhat}[3][0]{\altura=0 \advance\altura by #1
           \ancho=#2 \anchom=\ancho \divide\anchom by 2
           \anchoa=\ancho \divide\anchoa by 4
           \anchob=\anchom \advance\anchob by \anchoa
           \kern-3pt \begin{array}[b]{c}
           \begin{picture}(1,1)(\anchom,-\altura)
        \qbezier(0,2)(\anchoa,4)(\anchom,6)
        \qbezier(\anchom,6)(\anchob,4)(\ancho,2)
        \qbezier(0,2)(\anchoa,3.8)(\anchom,5.6)
        \qbezier(\anchom,5.6)(\anchob,3.8)(\ancho,2)
       \end{picture} \\[-4pt] {#3}
                       \end{array} \kern-4pt    }

\newcommand{\I}{I\mkern-7muI}

\makeatletter
\newcommand\prol{\@ifstar{\@proldf}{\@prolpf}}  %% if * dual else primal
\def\@prolpf{\@ifnextchar[{\@prolpf@wrt}{\@prolpf@}}
\def\@prolpf@wrt[#1]#2{\@ifnextchar[{\@prolpf@wrt@at{#1}{#2}}{\@prolpf@wrt@{#1}{#2}}}
\def\@prolpf@wrt@at#1#2[#3]{\prolsymbol^{#1}_{#3}#2}
\def\@prolpf@wrt@#1#2{\prolsymbol^{#1}#2}
\def\@prolpf@#1{\@ifnextchar[{\@prolpf@at{#1}}{\@prolpf@@{#1}}}
\def\@prolpf@at#1[#2]{\prolsymbol_{#2}#1}
\def\@prolpf@@#1{\prolsymbol#1}
\def\@proldf{\@ifnextchar[{\@proldf@wrt}{\@proldf@}}
\def\@proldf@wrt[#1]#2{\@ifnextchar[{\@proldf@wrt@at{#1}{#2}}{\@proldf@wrt@{#1}{#2}}}
\def\@proldf@wrt@at#1#2[#3]{\prolsymbol^{*#1}_{#3}#2}
\def\@proldf@wrt@#1#2{\prolsymbol^{*#1}#2}
\def\@proldf@#1{\@ifnextchar[{\@proldf@at{#1}}{\@proldf@@{#1}}}
\def\@proldf@at#1[#2]{\prolsymbol^*_{#2}#1}
\def\@proldf@@#1{\prolsymbol^*#1}
\def\prolsymbol{\mathcal{T}}
\makeatother

\newcommand{\Ht}{\mathbb{H}^2}
\newcommand{\htt}{\mathcal{H}^2}
\newcommand{\h}{\mathfrak{h}}
\newcommand{\g}{\mathfrak{g}}
\newcommand{\slt}{\mathfrak{sl}(2,\R)}
\newcommand{\sech}{\mbox{sech}}
\begin{document}
%{\Large

%\begin{tabular}{ll} {\bf Corresponding Author:} & Juan C Marrero \\ & e-mail: jcmarrer@ull.edu.es \\ & fax number: +34922318145
%\end{tabular}

\title[Relative equilibria for the two-body problem in the hyperbolic space of dimension 2]{Classification and stability of relative equilibria for the two-body problem in the hyperbolic space of dimension 2}

\author[L. C. \ Garc\'{\i}a-Naranjo]{Luis C. \ Garc\'{\i}a-Naranjo}
\address{L. C.\ Garc\'{\i}a-Naranjo:
Departamento de Matem\'aticas y Mec\'anica, IIMAS-UNAM \\
Apdo Postal 20-726 \\
Mexico City 01000, Mexico}
\email{luis@mym.iimas.unam.mx}

\author[J.\ C.\ Marrero]{Juan C.\ Marrero}
\address{Juan C.\ Marrero:
ULL-CSIC Geometr\'{\i}a Diferencial y Mec\'anica Geom\'etrica\\
Departamento de Matem\'aticas, Estad{\'\i}stica e IO, Secci\'on de
Ma\-te\-m\'a\-ti\-cas, Universidad de La Laguna, La Laguna,
Tenerife, Canary Islands, Spain} \email{jcmarrer@ull.edu.es}

\author[E. P\'erez-Chavela]{Ernesto P\'erez-Chavela}
\address{E. P\'erez-Chavela: Departamento de Matem\'aticas, Instituto Tecnol\'ogico Aut\'onomo de M\'exico (ITAM)\\
R\'io Hondo 1, Col. Pogreso Tizap\'an \\
 Mexico City 01080, Mexico}
\email{ernesto.perez@itam.mx}

\author[M. Rodr\'iguez-Olmos]{Miguel Rodr\'iguez-Olmos}
\address{M. Rodr\'iguez-Olmos: Departament de Matem\`atica\\
Universitat Polit\`ecnica de Catalunya\\
C. Jordi Girona, 31. 08034,
Barcelona, Spain}
\email{miguel.rodriguez.olmos@upc.edu}

%\thanks{This work has been partially supported by MEC (Spain)
%Grants MTM 2009-13383, MTM 2010-12116-E and the project of the Canary Government ProdID20100210.
%LGN Acknowledges the hospitality of the Departamento de Matem\'atica Fundamental,
%at Universidad de la Laguna, for its hospitality in numerous visits.
%}



\begin{abstract}
We classify and analyze the stability of all relative equilibria for the two-body problem in the hyperbolic space of dimension 2 and
we formulate our results in terms of the intrinsic Riemannian data of the problem.
\end{abstract}

\vspace{1cm}

\maketitle

\section{Introduction}\label{sec:intro}

This paper analyzes the relative equilibria of the 
generalization of the two-body problem to the complete, simply connected, hyperbolic, two-dimensional space $\htt$.
This is a continuation of the research program initiated in Diacu, P\'erez-Chavela, Santoprete \cite{Diac}, on
the study of relative equilibria for the $n$-body problem in spaces of constant curvature.
Here we focus on the simplest case of two particles in $\htt$. 
We present a rigorous
geometric treatment that allows us to fully classify and determine the stability properties of all relative equilibria.

This problem has a long history beginning with the formulation of the Kepler problem in spaces of constant curvature. 
Apparently (see  \cite{Florin1}), the analytic expression of the 
generalization of Newton's potential for the Kepler problem to the
three-dimensional hyperbolic  space $\mathcal{H}^3$ was first given 
by Schering \cite{Schering}, following the geometric ideas suggested decades earlier in the
works of Bolyai \cite{bolyai} and Lobachevsky \cite{loba}. The dynamics of this problem is integrable
and was thoroughly considered by Liebmann \cite{Liebmann}, and more recently by Kozlov et al \cite{kozlov} and Cari\~nena et al \cite{CaRaSa}.
Just like in the euclidean case, the bounded orbits are (hyperbolic) conics and a version of Kepler's  laws
holds. A nice explanation for this analogy in terms of projective geometrical arguments can be found in \cite{Albouy2}. 
We also mention that the ideas behind the generalization of the Kepler problem to constant curvature spaces have been taken
further into the subriemannian realm \cite{Mont-Shanbrom}. 
%This last reference contains a nice introduction that 
%explains why the first and second Kepler laws are valid in the hyperbolic space.
%

The formulation of the $n$-body problem in spaces of constant curvature follows by generalizing the Kepler potential
to include the pairwise interaction between all the masses \cite{kozlov}. In Section 3 of reference \cite{kozlov}, the authors 
indicate the interest in the investigation of its particular
solutions. An explicit expression of the equations
of motion can be found in \cite{Diac}. Physical and mathematical motivations to study this problem, as well as
some  historical details,  can be found  in \cite{Florin1}.

We remark that apart from the Kepler and the $n$-body problem, other mechanical systems with configuration manifold a space of constant curvature of dimension two have been discussed very recently, using differential geometric tools. This is the case, for instance, of  the two centre problem, the harmonic oscillator, and the $n$ vortex problem. See \cite{BM1, CaRaSa1,CaRaSa2,Montaldi-Nava,Voz} and the references therein.

Contrary to the euclidean case, the two-body problem in $\htt$
does not reduce to the Kepler problem and is non-integrable \cite{Sch2}. 
A number of publications
have addressed the dynamics of the  restricted two-body problem in $\htt$ (see \cite{BM2,Ki} and the
references therein). 
Just like the full two-body problem, it is nonintegrable
\cite{Mac-Prz}.

 The impossibility to reduce the two-body problem in $\htt$ to the Kepler problem
 has been associated to the absence of  ``the integral of the
center of mass" in  \cite{Florin1,FNDiacu}. From the geometric mechanics perspective, this seems to be related to the algebraic differences
between the group of orientation preserving isometries of euclidean two-space and of $\htt$ (see Remark \ref{R:Geom-Mech}). For a treatment of the  symmetry reduction of the problem see \cite{BM3,Sch1} (see also \cite{ScSt} and the references therein).

Previous results on the existence of relative equilibria for the two-body problem in $\htt$ first appeared
in \cite{Diac} in the case of equal masses and in \cite{Diacu2} in the general case. The results in  \cite{Diac}
are obtained in the Weierstrass model for $\htt$ that arises as an embedding of a hyperboloid in Minkowski's 3-dimensional space. The authors work with global coordinates in the ambient Minkowski space and thus refer to them
as {\em extrinsic coordinates}. On the other hand, the treatment in  \cite{Diacu2} is performed in 
the Poincar\'e disk and upper half-plane models. The authors use the terminology  {\em intrinsic coordinates}
and  they say that they take an  {\em intrinsic approach} since these spaces are not embedded in a larger ambient space. In 
this paper we will reserve the use of the term {\em intrinsic} to concepts that are intrinsic in the Riemannian geometry
sense, namely, that are invariant under isometries. In particular, contrary to the conventions taken in \cite{Diacu2}, in the terminology of the present paper, a choice of coordinates in a model of $\htt$, or  a mathematical formula expressed in these coordinates, are not intrinsic.

Our contribution is to classify all  relative equilibria of the problem and to fully describe their stability.
Moreover,  we formulate all of our results in an intrinsic form. In this way, all of our results and formulas 
(given in section  \ref{S:Intrinsic}) are valid in any model of $\htt$.

%We work in Poincar\'e's
%upper half-plane model $\Ht$ but we  present our
%results in terms of intrinsic Riemannian data. In this way, our results and formulas are valid in any model of $\htt$.
%

We show that the only relative equilibria of the problem arise as a conjugation
of the so-called elliptic and hyperbolic relative equilibria found in \cite{Diacu2}.  
 In order to study the problem of existence
for general configurations, we introduce a notion of {\em hyperbolic center of mass} between any two masses in $\htt$. Our definition seems to be
the correct one from a dynamical point of view since it allows us to simplify the study of existence and stability of relative equilibria
by assuming that the masses lie in a particularly simple configuration.
Interestingly,  a different definition of hyperbolic centroid of mass had been given before in the literature  \cite{Ga}.
% In fact, in \cite{Ga}, the center of mass of two material points is considered as a material point.

We remark that in the recent paper \cite{OrRe} the authors consider the existence problem of relative equilibria in the $n$-body problem (in particular, in the case $n = 2$) in the upper half plane of constant negative curvature. In this model, they also prove that, up to conjugation, the relative equilibria are elliptic and hyperbolic. However, they do not consider general configurations so their description is not exhaustive. In fact, in our paper, we prove that the study of existence may be reduced to the case of two masses in canonical configuration (see Proposition  \ref{P:Canonic-Conf}). This simplification is possible by the introduction of the concept of the hyperbolic center of mass and the exploitation of the homogeneity and isotropy of $\htt$.

%\textcolor{Red}{We remark that in a recent paper \cite{OrRe} the authors discuss the relative equilibria in the $n$-body problem (in particular, in the case $n=2$) in the upper half plane of constant negative curvature. In this model, they also prove that, up to the conjugation, the relative equilibria are elliptic and hyperbolic. However, they don't consider general configurations which implies that their description is not entirely complete as in our paper. In fact, in our paper, we prove that the study may be reduced to the case of two masses in canonical configuration (see Proposition \ref{P:Canonic-Conf}). This a consequence of the homogeneity and isotropy of the hyperbolic space of dimension two. Another key point in the theory is the use of our notion of the center hyperbolic of mass of two masses.}

Concerning the stability, we  show that all hyperbolic relative equilibria are unstable. A similar
result for a hyperbolic relative equilibrium of  three identical masses is given in \cite{Diac-Perez}.
However,  we show that the stability of  elliptic relative equilibria depends on the distance between the masses.
If the masses are sufficiently close, the relative equilibrium is stable. Otherwise the relative equilibrium is unstable. We give an explicit analytic expression that shows how the threshold
distance depends on the ratio between the masses and we also show that this distance is related to a critical value of the momentum. 

Our stability analysis is based on the Reduced-Energy-Momentum method of Simo et al \cite{Simo}, which gives sufficient conditions for a relative equilibrium to be $G_\mu$-nonlinearly stable in the sense of \cite{Patrick}. Several comments on this notion of stability are in order: First, this is a weaker notion than that of orbital stability, since we regard as $G_\mu$-stable relative equilibria for which nearby initial conditions can drift far away along $G_\mu$-orbits, the group $G_\mu$ being the subset of the symmetry group fixing the total momentum of the relative equilibrium. It follows from the geometric and dynamic nature of Hamiltonian systems with symmetry that to require orbital stability for a Hamiltonian relative equilibrium would be too restrictive and as a consequence,  $G_\mu$-stability is the common accepted notion of nonlinear stability in this field. An explicit example where a numerical simulation shows a relative equilibrium which is $G_\mu$-stable but not orbitally stable due to this drift can be found in \cite{LeoMa}.
Second, if a Hamiltonian relative equilibrium is $G_\mu$-stable then the projected equilibrium point to the symplectic reduced space corresponding to its momentum value  is Lyapounov stable in the usual sense. However the converse is not necessary true since lifting the conditions for Lyapounov stability on the reduced space to the original phase space would only take into account initial conditions with prescribed momentum. The precise relationship between these two notions of stability is clarified in \cite{Patrick}. It is worth noting that $G_\mu$-stability of a relative equilibrium also implies Lyapounov stability of the corresponding projected point in the reduced Poisson space. However, in this case the relationship between these two concepts of stability is more involved.  
 In our case,  two conclusions can be drawn about the solutions whose initial conditions are sufficiently close to a $G_\mu$-stable 
elliptic relative equilibrium. Firstly, they are bounded, and secondly, the variations in the distance between the particles is arbitrarily small during the
motion.

Our paper is organized as follows. In section \ref{prelim} we review known results on the hyperbolic space $\htt$ 
and of the Poincar\'e upper half-plane model $\Ht$. In section \ref{S:Geom-Formulation} we formulate the 2-body problem in $\htt$ and we give a useful   explicit formula for the potential
in the  $\Ht$-model.
 Section \ref{S:InfinitSymm} describes the symmetries of the system, its infinitesimal versions, and the corresponding conservation laws. In section \ref{S:Existence}
 we introduce the concept of hyperbolic center of mass  and  we prove all of our existence results.
 Section \ref{S:Stability} considers the stability of the relative equilibria found in section \ref{S:Existence}. Many of the results given in sections \ref{S:Geom-Formulation} --\ref{S:Stability} are 
formulated in the  $\Ht$-model and section \ref{S:Intrinsic} is concerned
with providing their intrinsic form and, as a consequence, the independence of these results on the chosen model. Finally, we include an Appendix which contains some standard results on $G_\mu$-stable relative equilibria for general symmetric Hamiltonian systems and 
for the $2$-body problem in an arbitrary Riemannian manifold.

% 
%
%working in the Weiestrass model of the hyperbolic geometry. Later Diacu et al gave an easier proof of the same fact using intrinsic coordinates in the model of the Poincar\'e disk \cite{Diacu2}.
%

%
%
%
%The $n$-body problem in spaces of constant curvature has its origins in N. Lovachevsky's formulation of the Kepler problem 
%in the three-dimensional hyperbolic space  \cite{loba}.
%
%
%This problem was first considered
% by... {\color{blue} The study of the two body problem, or more specifically, the Kepler problem in spaces of constant curvature started with the works of J. Bolyai \cite{bolyai} and N. Lovachevsky , the co-discovers of the first non Euclidean geometry, the last one defined a Kepler problem on the hyperbolic space $\mathbb{H}^3$ \cite{loba}. Since then many mathematicians, among them L. Dirichlet, E. Shering, R. Lipschitz, W. Killing, Liebmann, and more recently Cari\~nena, Ra\~nada, Santander \cite{cari},  the group of the Russians researchers led by V. Kozlov \cite{kozlov}, and Diacu, P\'erez-Chavela, Santoprete \cite{Diac, Diac2}. In \cite{Florin2} you can find a nice description of the contributions of the above people on this subject.}
% 
%{\color{red}Missing bolyai reference}
%
%

\section{Preliminaries}\label{prelim}

By the {\em  hyperbolic two-dimensional space} we understand the two-dimensional complete simply connected
Riemannian manifold with constant curvature $-1$. This is a well defined abstract space
that we denote by $\htt$. A well known
model for it is
the upper half plane $\Ht=\{ (x,y) \in \R^2 \, / \; \, y>0 \}$ equipped with the Riemannian metric
\begin{equation}
\label{E:HypMetric}
ds^2=\frac{dx^2 +dy^2}{y^2}.
\end{equation}

In this work, we will mainly work with this model but we will formulate our results in such way that they
are valid in the abstract hyperbolic two-dimensional space $\htt$. Our choice to work in $\Ht$ stems from three nice properties.
The first one is that it can be conveniently covered with a single coordinate chart $(x,y)$. The second one is that this model is conformal,
in the sense that the angle between any two curves in the upper half plane that intersect is the same when measured with the
 $\Ht$ metric and with the euclidean metric. This property suggests the identification $T_{(x,y)}\Ht\cong \R^2$ that
 we will use in the sequel.
 
The third convenient property  is that $\Ht$ can be endowed
 with a Lie group structure arising from the following matrix representation:
\begin{equation*}
(x,y) \mapsto \left ( \begin{array}{cc} y & x \\ 0 & 1 \end{array} \right ).
\end{equation*}
This Lie group can be interpreted as the semi-direct product of $(\R^+, \cdot)$ and $(\R,+)$.
That is, the semi-direct product of expansions and translations on $\R$. The Lie algebra $\h$ is
spanned by the matrices
\begin{equation*}
e_1:=\left ( \begin{array}{cc}  0 & 1 \\ 0 & 0  \end{array} \right ), \qquad e_2:=\left ( \begin{array}{cc}  1 & 0 \\ 0 & 0  \end{array} \right ),
\end{equation*}
with commutation relation
\begin{equation*}
[e_1,e_2]=-e_1.
\end{equation*}

A crucial observation is that the hyperbolic metric \eqref{E:HypMetric} is invariant under the left multiplication on $\Ht$.
Therefore $\Ht$ is a subgroup of the group of former isometries of $\Ht$. This implies that the isometry group of $\Ht$ acts transitively which means that {\em the hyperbolic two-dimensional space $\htt$ is homogenous}.

%
%
%For the sequel it is useful to give a basis of left (right) invariant vector fields and differential forms.
%They are obtained by translation of the basis of $\h$ and its dual.
%\begin{equation}
%\label{E:SpecialVFinH2}
%\begin{split}
%&X_L=y\partial_x, \qquad Y_L=y\partial_y, \qquad \qquad \alpha_L=\frac{dx}{y}, \qquad \beta_L=\frac{dy}{y}, \\
%&X_R=\partial_x, \qquad Y_R=x\partial_x+ y\partial_y, \qquad \qquad \alpha_R=dx-x\frac{dy}{y}, \qquad \beta_R=\frac{dy}{y}.
%\end{split}
%\end{equation}
%
%The flow of the right invariant vector fields leaves the metric invariant. That is, they are isometries.
%Such flow is defined by
%\begin{equation*}
%\begin{split}
%\Phi_t: (x,y) \mapsto (x+t, y) \qquad \mbox{(flow of $X_R$)} \\
%\Phi_t: (x,y) \mapsto (e^tx, e^ty) \qquad \mbox{(flow of $Y_R$)}.
%\end{split}
%\end{equation*}
%These transformations are translations along the $x$ axis and homothesis centered at $(0,0)$.

The following, well-known result, follows from
 Hadamard theorem (see, for instance, Chapter 10, Theorem 22 in \cite{On}).

\begin{proposition}
\label{P:PropHypSpace1}
Given any two points in
$\htt$ there exists a unique geodesic passing through them.
\end{proposition}

The following proposition is also a well-known result (see, for instance, Chapter 3, Example 3.10 in \cite{Ca})

\begin{proposition}
\label{P:PropHypSpace2}
 The geodesic curves on $\Ht$ are either half circumferences
cutting at right angles the $x$-axis $y=0$, or vertical lines $x=const$.
\end{proposition}
%
%{\color{blue} Let $\ell: \R \to \Ht$ be a half circunference cutting at right angles the $x$-axis. Consider the tangent vector $\dot{\ell}(t)$ at the geodesic $\ell$ and the subspace $<\dot{\ell}(t)>^{\perp}$ of $T_{\ell(t)}\Ht \simeq \R^2$ which consists of the orthogonal vectors to $\dot{\ell}(t)$. Then, it is easy to see that if $v = (v_1, v_2) \in <\dot{\ell}(t)>^{\perp}$
%\[
%v \neq 0 \Leftrightarrow v_2 \neq 0.
%\]
%Now, if $v = (v_1, v_2) \in <\dot{\ell}(t)>^{\perp}$, $v' = (v'_{1}, v'_{2}) \in <\dot{\ell}(t')>^{\perp}$ and $v, v' \neq 0$ then $v$ and $v'$ are said to be equally oriented (respectively, oppositely oriented) if $v_2$ and $v'_2$ have the same sign (respectively, have different sign).
%
%On the other hand, if $\ell: \R \to \Ht$ is a vertical line and $v = (v_1, v_2) \in <\dot{\ell}(t)>^{\perp}$ it follows that
%\[
%v \neq 0 \Leftrightarrow v_1 \neq 0.
%\]
%Thus, $v = (v_1, v_2) \in <\dot{\ell}(t)>^{\perp}$ and $v' = (v'_1, v'_2) \in <\dot{\ell}(t')>^{\perp}$ are said to be equally oriented (respectively, oppositely oriented) if $v_1$ and $v'_1$ have the same sign (respectively, have different sign).
%
%Therefore, if $\ell: \R \to \Ht$ is a geodesic in $\Ht$ then one may introduce the notion of equally oriented (respectively, oppositely oriented vectors) $v$ and $v'$, with $v \in <\dot{\ell}(t)>^{\perp}$, $v' \in <\dot{\ell}(t')>^{\perp}$ and $v, v' \neq 0$.
%
%In fact, this notion may be transferred, in a natural way, to the abstract hyperbolic space $\htt$.
%}
%{\color{red}
Consider a parametrized geodesic $\ell:\R\to \Ht$ and let $v_1, v_2$ be tangent vectors to $\Ht$ at $\ell(t_1)$, $\ell(t_2)$, for some
$t_1, t_2\in \R$. Suppose moreover, that   $v_i$ is perpendicular to $\dot \ell (t_i)$, $i=1,2$. We say that the vectors $v_1$ and $v_2$ are {\em equally oriented}
if the bases $\{v_1, \dot \ell(t_1)\}$ and $\{v_2, \dot \ell(t_2)\}$ define the same orientation of $\R^2$. A pair of  vectors that is not
 equally oriented is said to be  {\em oppositely oriented}.

Therefore, if $\ell\subset \Ht$ is a geodesic in $\Ht$,  
one may introduce the notions of equally and oppositely oriented 
non-zero tangent vectors $v_1$ and $v_2$ based at different points in $\ell$ and that are 
perpendicular to $\ell$ (see Figures \ref{F:eq-op-orient} and \ref{F:eq-op-orient-vert}   for an illustration). 
The correctness of the
definition relies on the possibility to parametrize geodesics, and on a choice of orientation of $\Ht$. Hence, 
 this concept may be transferred, in a natural way, to the abstract hyperbolic space $\htt$.
 %}
This notion will be useful in the intrinsic description of relative equilibria in section \ref{S:Intrinsic}.

\begin{figure}[h]

\centering
\subfigure[Equally oriented.]{\includegraphics[width=6cm]{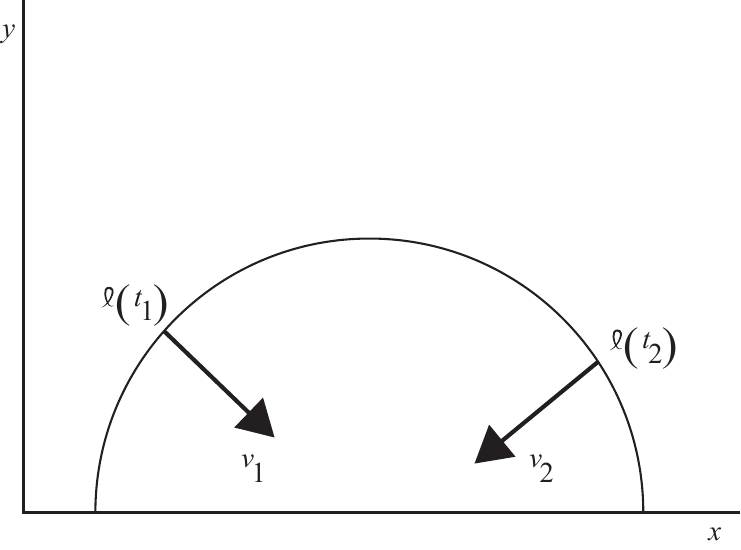}} \qquad 
\subfigure[Oppositely oriented.]{\includegraphics[width=6cm]{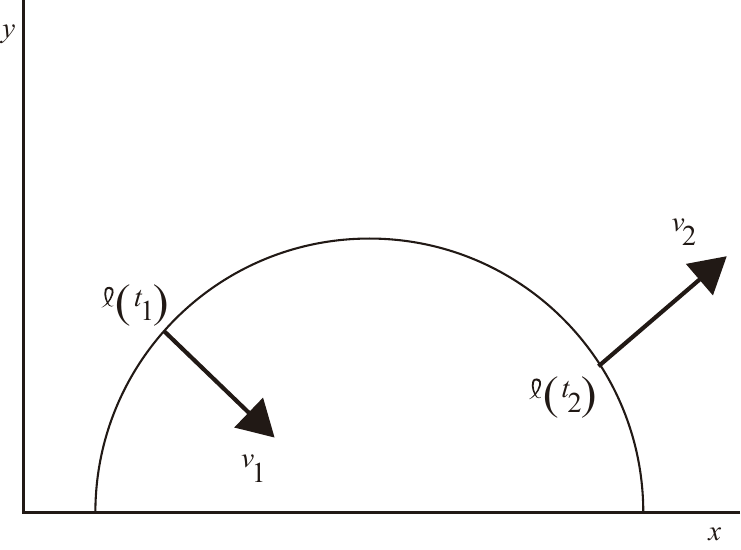}}
\caption{Illustration of equally and oppositely oriented vectors for a half circle geodesic in $\Ht$.} \label{F:eq-op-orient}
\end{figure}

\begin{figure}[h]

\centering
\subfigure[Equally oriented.]{\includegraphics[width=6cm]{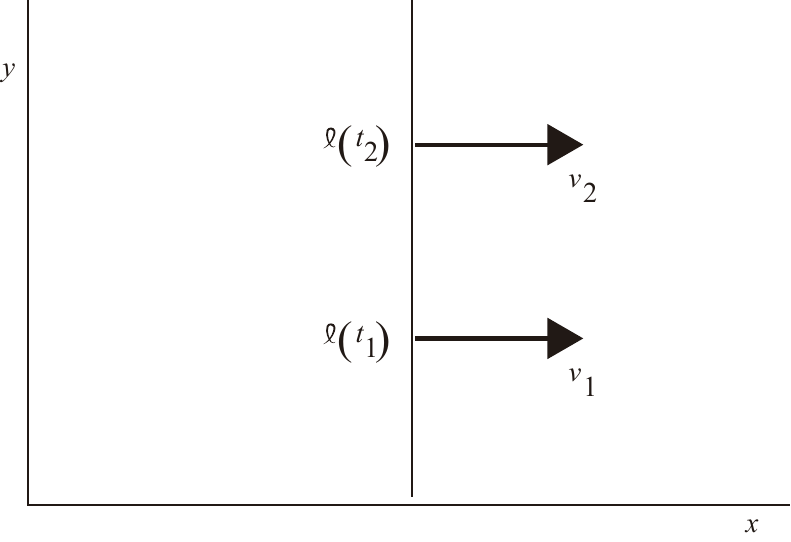}} \qquad 
\subfigure[Oppositely oriented.]{\includegraphics[width=6cm]{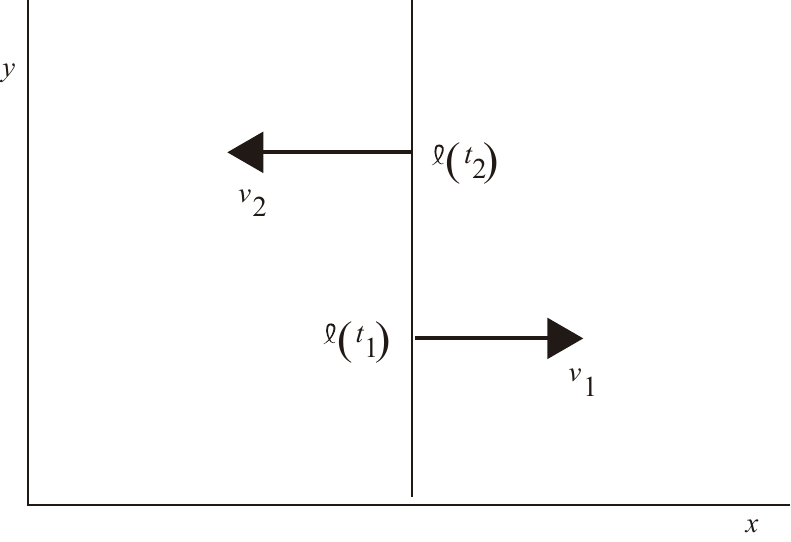}}
\caption{ Illustration of equally and oppositely oriented vectors for a vertical geodesic in $\Ht$.} \label{F:eq-op-orient-vert}
\end{figure}

For future reference, we note that the Riemannian distance,
$d$, between two points $(x_1,y_1), \, (x_2, y_2)\in \Ht$ satisfies
\begin{equation}
\label{E:distanceFmla}
\cosh(d)=1+\frac{(x_1-x_2)^2+(y_1-y_2)^2}{2y_1y_2}.
\end{equation}
In particular, if $(x_1, y_1) = (\pm\cos\theta, \sin \theta)$, with $0 < \theta < \frac{\pi}{2}$, and $(x_2, y_2) = (0, 1)$ then
\begin{equation}\label{trigon-hyper}
\sin \theta = \sech (d), \; \; \; \cos \theta = \tanh (d).
\end{equation}

\subsection{The group of isometries of $\Ht$}
\label{S:Moebius}

Let
\[ {\rm SL}(2,\mathbb{R}) = \{ g \in {\rm GL}(2,\mathbb{R}) \, | \, \det g=1 \},  \]
where ${\rm GL}(2,\mathbb{R})$ is the group of invertible $2 \times 2$ real matrices,
be the {\it special linear real group}, which is a 3-dimensional simply connected,  real Lie group.

The group $ {\rm SL}(2,\mathbb{R}) $ defines an action $\Psi$ on $\Ht$ as follows.
Given a matrix $g \in {\rm SL}(2,\mathbb{R})$,
\begin{equation*}
g = \left ( \begin{array}{cc} a & b \\ c & d \end{array} \right ),
\end{equation*}
and an element $(x,y)\in \Ht$,  we define $\Psi(g,(x,y))$ via  the  M\"obius transformation
\begin{equation}
\label{E:DefAction}
\Psi(g,(x+iy))= \frac{a(x+iy) + b}{c(x+iy) +d},
\end{equation}
where we have used complex notation for elements in $\R^2$.
One can check that  the upper half plane $y>0$ is invariant under the above transformation and that
the action axioms are verified, so $\Psi$ is well defined.

It is well known  that  for any $g\in{\rm SL}(2,\mathbb{R}) $,  the map $\Psi(g,\cdot)$ is an orientation preserving Riemannian isometry of
$\Ht$ equipped with the hyperbolic metric. Moreover, the group of orientation preserving, proper isometries of $\mathbb{H}^2$ is the quotient
group $\displaystyle {\rm SL}(2,\mathbb{R})/\{\pm I \}:={\rm PSL}(2,\mathbb{R})$. The quotient is taken to account for the  fact that $\Psi(g,\cdot)=\Psi(-g,\cdot)$. In  the sequel we will sometimes work with the group ${\rm SL}(2,\mathbb{R})$ keeping in mind that
isometries are counted twice.

%A basis for this Lie algebra is given by
%\[ \left\{
%X_1 = \frac{1}{2} \left( \begin{array}{ccc}
%    1 & 0 \\
%    0 & -1  \\
%    \end{array}\right),  \quad
%X_2 = \left(\begin{array}{cc}
%     0 & 1 \\
%    0 & 0  \\
%    \end{array}\right), \quad
%X_3 = \left(\begin{array}{cc}
%    0 & 1 \\
%    -1 & 0  \\
%    \end{array}\right).
%\right\} \]

As mentioned before,  a good number of isometries defined by elements in ${\rm SL}(2,\R)$ correspond to
left translation on $\Ht$. Indeed, left multiplication by the element
\begin{equation*}
\left ( \begin{array}{cc} y_0 & x_0 \\ 0 & 1 \end{array} \right ) \in \Ht
\end{equation*}
corresponds to the action of the matrix
\begin{equation*}
\left ( \begin{array}{cc} \sqrt{y_0} & \frac{x_0}{\sqrt{y_0}} \\ 0 & \frac{1}{\sqrt{y_0}} \end{array} \right ) \in
{\rm SL}(2,\R).
\end{equation*}
In fact, we can consider $\Ht$ as a Lie subgroup of ${\rm SL}(2,\R)$ via the group
homomorphism
\begin{equation}
\label{E:H2subgroupofSL2}
\left ( \begin{array}{cc} y_0 & x_0 \\ 0 & 1 \end{array} \right )  \mapsto \left ( \begin{array}{cc} \sqrt{y_0} & \frac{x_0}{\sqrt{y_0}} \\ 0 & \frac{1}{\sqrt{y_0}} \end{array} \right ).
\end{equation}

The action $\Psi$ of  ${\rm SL(2,\R)}$ on $\Ht$ introduces a new type of symmetries that do not arise as translations on $\Ht$.
These transformations are analogous to usual rotations on the plane. This family of symmetries is generated
by an ${\rm S}^1$ subgroup of ${\rm PSL}(2,\R)$ consisting of matrices
of the form
\begin{equation*}
g_t:= \left ( \begin{array}{cc} \cos \left ( \frac{t}{2} \right ) & -\sin \left ( \frac{t}{2} \right )\\ \sin \left ( \frac{t}{2} \right ) &
 \cos \left ( \frac{t}{2} \right ) \end{array} \right ), \qquad t\in \R \quad \mbox{mod} \, 2\pi.
\end{equation*}
The corresponding transformation on $\Ht$ acts like a rotation by the angle $t$ around $(0,1)$ in the following way:
it fixes the identity element $(x=0, y=1)$  and takes the geodesic through $(0,1)$ with tangent vector $v\in \R^2\cong T_{(0,1)}\Ht$ into the
the geodesic through $(0,1)$ with tangent vector $g_{-2t}v\in \R^2\cong T_{(0,1)}\Ht$ .
Thus, this subgroup acts transitively in the space of directions through $(0,1)$. This, together with the homogeneity of space,
shows that {\em the hyperbolic two-dimensional space $\htt$ is isotropic}.

 We shall call the above
transformations ``hyperbolic rotations" around $(0,1)$.

It is easy to prove that the hyperbolic rotations around $(0,1)$ are just the isotropy subgroup ${\rm SL(2,\R)}_{(0,1)}$ of $(0,1)$ with respect to the ${\rm SL(2,\R)}$-action on $\Ht$. Thus, using Theorem A.1.1 in \cite{Ku}, the compactness of ${\rm SL(2,\R)}_{(0,1)}$ and the fact that the action of ${\rm SL(2,\R)}$ on $\Ht$ is transitive, we deduce an important result: the action of ${\rm SL}(2, \R)$ on $\Ht$ is proper.  

\subsection{A property of isometries in $\Ht$}

In this section, we present a basic result about isometries of $\Ht$ that will be useful in the sequel.

\begin{proposition}
\label{P:Geod-unitcircle}
Let $\ell \subset \Ht$ be a complete geodesic curve and let $(x_0,y_0)\in \ell$. There exists an isometry of $\Ht$ that maps
$\ell$ onto the upper half of the unit circle and $(x_0,y_0)$ onto $(0,1)$.
\end{proposition}
\begin{proof}
By homogeneity of $\Ht$ there exists an isometry defined by $g_1\in {\rm SL}(2,\R)$ that maps $(x_0,y_0)$ onto $(0,1)$. Since isometries
take geodesics into geodesics, the image of   $\ell$ by this isometry is a certain geodesic $ \tilde \ell$ through  $(0,1)$.
Next, by isotropy of $\Ht$  there exists a hyperbolic rotation around $(0,1)$ defined by $g_2\in {\rm SL}(2,\R)$ that maps $ \tilde \ell$
onto the upper half of the unit circle. The element $g_2g_1\in  {\rm SL}(2,\R)$ defines the desired isometry.
\end{proof}
An explicit formula can be given for the isometry described in the statement of the above proposition. Let $v\in T_{(x_0,y_0)}\Ht$ be any non-zero vector tangent to $\ell$. It can be written as
\begin{equation*}
v=\lambda y_0 (\cos \theta_0 , \sin \theta_0) \in \R^2,
\end{equation*}
for a certain $\theta_0\in [0,2\pi)$, where $\lambda>0$ is the hyperbolic norm of $v$. The isometry defined by
\begin{equation*}
g_0=\left ( \begin{array}{cc} a & b \\ c & d \end{array} \right )\in {\rm SL}(2,\R),
\end{equation*}
where
\begin{equation*}
\begin{split}
a=\frac{1}{\sqrt{y_0}}\cos \left ( \frac{\theta_0}{2} \right ), \quad b=-\frac{\cos  \left ( \frac{\theta_0}{2} \right ) x_0 +
\sin  \left ( \frac{\theta_0}{2} \right ) y_0}{\sqrt{y_0}}, \\ c=\frac{1}{\sqrt{y_0}}\sin \left ( \frac{\theta_0}{2} \right ),
\quad d=\frac{-\sin  \left ( \frac{\theta_0}{2} \right ) x_0 +
\cos  \left ( \frac{\theta_0}{2} \right ) y_0}{\sqrt{y_0}},
\end{split}
\end{equation*}
satisfies the requirements of Proposition \ref{P:Geod-unitcircle}.

Note that $g_0 = g_2 g_1$, with $g_1, g_2 \in {\rm SL}(2, \R)$ given by
\[
g_1 = \left ( \begin{array}{cc} \frac{1}{\sqrt{y_0}} & -\frac{x_0}{\sqrt{y_0}} \\ 0 & \sqrt{y_0} \end{array} \right ),
\; \; \;
g_2 = \left ( \begin{array}{cc} \cos \left ( \frac{\theta_0}{2} \right ) & -\sin \left ( \frac{\theta_0}{2} \right )\\ \sin \left ( \frac{\theta_0}{2} \right ) &
 \cos \left ( \frac{\theta_0}{2} \right ) \end{array} \right ).
 \]

\section{Geometric formulation of the 2--body problem in  $\htt$}
\label{S:Geom-Formulation}

The  $2$--body problem in $\htt$ can be defined as the simple mechanical system on the configuration space
$Q=(\htt)^2\setminus\Delta$, where $\Delta$ is the set of collision configurations, and whose Lagrangian
$L:TQ\to \R$ is given by
\begin{equation*}
L=\frac{1}{2}\left ( m_1v^2_1+m_2v_2^2 \right ) -  V(d).
\end{equation*}

In the above equation $m_i$ and  $v_i^2$ are, respectively, the mass and the square of the hyperbolic norm
of the velocity of the $i^{th}$ particle. The positive number $d$ is the hyperbolic distance between the first and
the second mass and the potential $V:\R^+\to \R$ is given by
\begin{equation}
\label{potential}
V(d)=-km_1m_2\coth(d),
\end{equation}
where $k$ is a positive constant. The hyperbolic cotangent dependence of $V$ on the distance  is well-accepted to be the ``correct"
 generalization of the Newtonian gravitational potential for two reasons. The first one is that  the Kepler problem with this potential  satisfies Bertrand's property, namely, all of the bounded orbits are closed.
The second reason is that, just like the   Newtonian
potential $V_N(d)=\frac{1}{d}$ is (proportional to) the fundamental solution of the Laplacian operator on $\R^3$,
the proposed potential $V(d)$ is (proportional to) the fundamental solution of the Laplace-Beltrami operator on
 $\mathcal{H}^3$. As a consequence, it satisfies the hyperbolic version of Gauss law.

% On the other hand, a similar potential was considered in \cite{CaRaSa} for the Kepler problem in the hyperbolic space of negative constant curvature.

In what follows we will work in the $\Ht$ model. We will also denote $Q=(\Ht)^2\setminus(\Delta)$ with elements $q=(x_1,y_1,x_2,y_2)$.
 The Hamiltonian
of the problem is $h:T^*Q\to \R$ given by
\begin{equation*}
h(p_q)=\frac{1}{2}||p_q||^2+V(q)
\end{equation*}
where $||\cdot ||$ denotes the norm in the fibers of $T^*Q$ induced by the kinetic energy Riemannian metric on $Q$ given by
\begin{equation}
\label{E:Metric-on-Q}
m_1\frac{dx_1^2+dy_1^2}{y_1^2}+m_2\frac{dx_2^2+dy_2^2}{y_2^2}.
\end{equation}
%
%
%
%We now proceed to write explicit formulae for the equations of motion for the $2$--body problem in the $\Ht$ model.
%We will also denote $Q=(\Ht)^2\setminus(\Delta)$. The equations in $TQ$ are written in terms of the global, cartesian coordinates.
%

Using \eqref{E:distanceFmla}
\begin{equation*}
\label{E:Hyperb-distance}
d=\cosh^{-1}\left ( 1 + \frac{(x_1-x_2)^2+(y_1-y_2)^2}{2y_1y_2} \right )
\end{equation*}
and the relationship
\begin{equation*}
\coth(\cosh^{-1}(z))=\frac{z}{\sqrt{z^2-1}}
\end{equation*}
that is valid for all $z>1$ we can write
\begin{equation}
\label{E:Potential}
V(x_1,y_1,x_2,y_2)=-km_1m_2\frac{(x_1-x_2)^2+y_1^2+y_2^2}{\sqrt{((x_1-x_2)^2+(y_1-y_2)^2)((x_1-x_2)^2+(y_1+y_2)^2)}}.
\end{equation}
%and
%\begin{equation}\label{E:Lagrangian}
%\begin{split}
%L &= \frac{1}{2}\left ( \sum_{i=1}^2 m_i \frac{\dot x_i^2+\dot y_i^2}{y_i^2} \right )
%% \\ &\qquad  \qquad \qquad
% +k m_1m_2\frac{(x_1-x_2)^2+y_1^2+y_2^2}{\sqrt{((x_1-x_2)^2+(y_1-y_2)^2)((x_1-x_2)^2+(y_1+y_2)^2)}}.
%\end{split}
%\end{equation}
%
%The Euler-Lagrange equations for the above Lagrangian can be cast in the form:
%\begin{equation}
%\label{E:Motion-n-equal2}
%\begin{split}
%\ddot x_i &= \frac{2\dot x_i \dot y_i}{y_i} +  \sum_{j\neq i}\frac{km_j8(x_j-x_i)y_i^4y_j^2}{[((x_i-x_j)^2+(y_i-y_j)^2)((x_i-x_j)^2+(y_i+y_j)^2)]^{3/2}} \, , \\
%\ddot y_i &= \frac{\dot y_i^2 -\dot x_i^2}{y_i} +  \sum_{j\neq i}\frac{km_j4y_i^3y_j^2((x_i-x_j)^2+y_j^2-y_i^2)}{[((x_i-x_j)^2+(y_i-y_j)^2)((x_i-x_j)^2+(y_i+y_j)^2)]^{3/2}},
%\end{split}
%\end{equation}
%for $i=1,2$.
%The first term appearing in the right of the above equations is the inertia of the point mass $m_i$ manifesting its
%tendency to move along a geodesic of $\Ht$. The second term is the force due to the gravitational interaction with  the
%other mass.
%
%The configuration space  for the system is $Q=(\Ht)^N \setminus \Delta$ where $\Delta\subset (\Ht)^{N}$ is the collision set.
% The Lagrangian $L$ is a smooth function on $TQ$.

%The kinetic energy of the Lagrangian defines a Riemannian metric on $Q$. For every $q\in Q$, the metric defines a scalar product in
%$T_qQ$ that  will be denoted by $\langle\langle \cdot , \cdot \rangle \rangle_q$.
%
%
\section{Infinitesimal symmetries and the momentum map}
\label{S:InfinitSymm}

Since the potential energy only depends on the hyperbolic distance between the
particles, it follows that the diagonal
action of $\mbox{SL}(2,\R)$  on $Q$ lifts to a symmetry of the problem. 
It is well known that  nontrivial isometries of $\htt$ that fix more than one point do not preserve orientation. Therefore,
since we have removed the collision
configurations, this action is free. 

In addition, using the proper character of the action of ${\rm SL(2, \R)}$ on $\Ht$, we directly deduce that the action of ${\rm SL(2,\R)}$ on $Q$ also is proper.

\begin{remark}{Since the action of ${\rm SL(2, \R)}$ on $Q$ is free and proper, we have that $Q$ is the total space of a principal ${\rm SL(2, \R)}$-bundle over the space of orbits $Q/{\rm SL(2, \R)}$. In fact, using  Proposition \ref{P:Geod-unitcircle} and formula \eqref{trigon-hyper}, one may prove that $Q/{\rm SL(2, \R)}$ is diffeomorphic to $\R^+=\{r\in \R :  r>0 \}$. Moreover, under this diffeomorphism, 
the orbit projection $\pi:Q\to \R^+$ is just the restriction to $Q$ of the hyperbolic Riemannian distance.}
\end{remark}
In this section we write the infinitesimal version of the ${\rm SL(2, \R)}$-symmetry in the $\Ht$ model, we make a classification
of the Lie algebra elements that will be useful for the sequel and we compute the
corresponding momentum map.

The  Lie algebra of both ${\rm PSL}(2,\mathbb{R})$ and ${\rm SL}(2,\mathbb{R})$ is the 3-dimensional real linear space
\[ \slt = \{ \xi \in {\rm M}_{2\times 2}(\mathbb{R}) \, | \, \,  {\rm trace\  \! \xi}=0  \}.  \]

Hence, any non-zero
element in $\slt$ is of one of the following three types:

\begin{enumerate}
\item Elliptic. These elements have two complex conjugate, purely imaginary eigenvalues.
\item Hyperbolic. They posses two real eigenvalues with the same absolute value and \\ opposite signs.
\item Parabolic. They have a multiplicity two zero eigenvalue  (and are not diagonalizable).
\end{enumerate}

The following are representatives of elements of the above types and form a basis of $\slt$.
\begin{equation*}
\xi_e:=\left ( \begin{array}{cc} 0 & -\frac{1}{2} \\ \frac{1}{2} & 0 \end{array} \right ), \qquad
\xi_h:=\left ( \begin{array}{cc} \frac{1}{2} &0 \\ 0 & -\frac{1}{2} \end{array} \right ), \qquad \xi_p:=\left ( \begin{array}{cc} 0 & 1 \\ 0 & 0 \end{array} \right ).
\end{equation*}
We have that
\begin{equation}\label{structure-constants}
[\xi_e, \xi_h] = \xi_e + \xi_p, \; \; [\xi_e, \xi_p] = -\xi_h, \; \; \; [\xi_h, \xi_p] = \xi_p.
\end{equation}
 Let $\xi$ be an arbitrary elliptic element in $\slt$ with eigenvalues $\pm i\frac{\omega}{2}$, $\omega\in \R$,
 and corresponding
complex eigenvectors $u\pm iv$ with $u,v\in \R^2$. We then have
\begin{equation}
\label{E:AuxElliptic}
\xi u=-\frac{\omega}{2} v, \qquad \xi v=\frac{\omega}{2} u.
\end{equation}
From these relations  it is easy to show that $u$ and $v$ are linearly independent over
$\R^2$. Suppose that $\{u,v\}$ is positively oriented. Then one can scale $u$ and $v$ by the same positive factor
so that \eqref{E:AuxElliptic} continues to hold and the matrix $g=:(u|v)$, that has
$u$ and $v$ as columns, has determinant
one so it belongs to $ {\rm SL}(2, \R)$. In view of \eqref{E:AuxElliptic}
we obtain
\begin{equation*}
{\rm Ad}_{g^{-1}}\xi =g^{-1}\xi g=-\omega \xi_e.
\end{equation*}
If $\{u,v\}$ is not positively oriented, after a rescaling by a positive factor, the matrix $g=(u|-v)\in {\rm SL}(2, \R)$ and
satisfies ${\rm Ad}_{g^{-1}}\xi =\omega \xi_e$.
In a similar way, given a parabolic element  $\xi \in \slt$, one can find a basis $\{u,v\}$ of $\R^2$ such that
\begin{equation*}
\xi u=0, \qquad \xi v=u.
\end{equation*}
If $\{u,v\}$ is positively oriented then we can scale $u$ and $v$  by the same
 positive factor so that the determinant of the matrix $g=(u|v)$ is $1$, and the above relations still hold.
 In this case
 \[
 \mbox{ Ad}_{g^{-1}}\xi =\xi_p.
\]
 If $\{u,v\}$ is not positively oriented one can rescale  by a positive factor so that the matrix $g=(u|-v)\in {\rm SL}(2,\R)$
 and we have $ \mbox{ Ad}_{g^{-1}}\xi =-\xi_p$.
 
Finally,  any   hyperbolic element in $\xi \in \slt$ with non-zero real eigenvalues $\pm \frac{\omega}{2}$ is 
diagonalizable. Rescaling the eigenvectors (by positive or negative factors) shows that one may
choose $g \in {\rm SL}(2, \R)$ such that
 \[
\xi = \omega g \xi_h g^{-1} = \omega {\rm Ad}_{g}\xi_e.
\]
%}
In conclusion, we have proved the following.
\begin{proposition}\label{ellip-hyper-para}
Let $\xi$ be an element in $\slt$. Then:
\begin{enumerate}
\item
$\xi$ is elliptic (respectively, hyperbolic) if and only if there exists $g \in {\rm SL}(2, \R)$ and $\omega \in \R$, $\omega \neq 0$ such that
\[
\xi = \omega {\rm Ad}_g \xi_e
\]
(respectively, $\xi = \omega  {\rm Ad}_g \xi_h$).
\item
$\xi$ is parabolic if and only if there exists $g \in {\rm SL}(2, \R)$ such that
\[
\xi = \pm  {\rm Ad}_g \xi_p.
\]
\end{enumerate}
\end{proposition}

%In this way we obtain a complete description of the adjoint orbits of ${\rm SL}(2,\R)$ on $\slt$. Representatives of the orbits are
%$\omega \xi_e, \, \omega \xi_h, \xi_p^+$ and $\xi_p^-$, where $\omega$ is a positive real parameter.
%

The description of $\slt$ in terms of elliptic, hyperbolic and parabolic elements ties in very well with our study of ${\rm SL}(2,\R)$ done in section \ref{S:Moebius}. By a direct calculation, one can see that elements $\xi_h$ and $\xi_p$ span
a two-dimensional Lie subalgebra of $\slt$. The underlying Lie subgroup of ${\rm SL}(2,\R)$ is isomorphic to $\Ht$ via the map \eqref{E:H2subgroupofSL2}.

Starting from the definition of the action $\Psi$ in \eqref{E:DefAction}, a direct calculation
shows that
\begin{equation}
\label{E:Hyp-Par}
\Psi \left (\exp(\omega \xi_h t), ( x,y)\right ) = e^{\omega t}
( x,y),
\qquad
\Psi \left (\exp(\omega \xi_p t), ( x,y ) \right ) =
\left (  x+\omega t ,y \right ).
\end{equation}
%Note that these are precisely the flows of $\omega X_R$ and $\pm\omega Y_R$ where $X_R$ and $Y_R$ are given in \eqref{E:SpecialVFinH2}.

On the other hand, one has
\begin{equation}
\label{E:Ellip}
\Psi \left (\exp(\omega\xi_e t), ( x,y ) \right )= \frac{\left ( \sin (\omega t)(x^2+y^2-1)+2x\cos(\omega t) \, , \, 2y \right )}{(x^2+y^2+1)+(1-x^2-y^2)\cos(\omega t)  +2 x\sin (\omega t) },
\end{equation}
which is a hyperbolic rotation around $(x,y)$ as introduced in Section  \ref{S:Moebius}. The terminology
is somewhat unfortunate since ``elliptic" elements in the Lie algebra $\slt$ generate ``hyperbolic"
rotations.

Putting $\omega=1$, differentiating  and evaluating at $t=0$ formulas \eqref{E:Hyp-Par} and \eqref{E:Ellip}, we obtain
the following expressions for the infinitesimal generators of the action
\begin{equation}
\label{E:Infinit-Generator1}
\begin{split}
(\xi_h)_{\Ht}(x,y)&= ( x, y)\in T_{(x,y)}\Ht, \qquad (\xi_p)_{\Ht}(x,y)=(1 ,0)\in T_{(x,y)}\Ht ,\\
 (\xi_e)_{\Ht}(x,y)&= \left ( \frac{y^2-x^2-1}{2},- xy \right )\in T_{(x,y)}\Ht.
\end{split}
\end{equation}
The corresponding expressions for the infinitesimal generators of the diagonal action of ${\rm SL}(2,\R)$
on $Q$ are readily obtained from the above formulas.

In our treatment we will identify the dual Lie algebra $\slt ^*$ with $\slt$ via the pairing
\begin{equation}
\label{E:Dual-identification}
\langle \mu , \xi \rangle = 2\mbox{Trace}(\mu \xi), \qquad \mbox{for} \qquad \mu\in \slt ^*\cong \slt, \; \xi \in \slt.
\end{equation}
The basis of $\slt^*\cong \slt$ that is dual to $\{ \xi_h, \xi_e, \xi_p\}$ via this pairing is $\{ \mu_h, \mu_e, \mu_p\}$ with
\begin{equation}
\label{E:dualbasis}
\mu_h=\xi_h, \qquad \mu_e=\xi_p,  \qquad \mu_p=(\xi_p+\xi_e).
\end{equation}
In addition, with the
previous identification
\begin{equation}
\label{E:adstar}
{\rm ad}^*_\xi\mu=[\mu,\xi], \qquad {\rm and} \qquad {\rm Ad}^*_{g^{-1}}\mu=g\mu g^{-1}.
\end{equation}

We can now compute the momentum map associated with this symmetry. The momentum map of the cotangent lift of the diagonal action
of ${\rm SL}(2,\R)$ on $Q$ is  $J:T^*Q\to \slt ^*\cong \slt$ given by
 \begin{equation}
 \label{E:J}
\begin{split}
J(x_1,x_2,y_1,y_2,p_{x_1},p_{x_2}, p_{y_1},  p_{y_2})&= \left ( \sum_{i=1}^2x_ip_{x_i}+y_ip_{y_i} \right ) \mu_h
+\left ( \sum_{i=1}^2\frac{p_{x_i}(y_i^2-x_i^2-1)-2p_{y_i}x_iy_i}{2} \right ) \mu_e \\
& \qquad + \left ( \sum_{i=1}^2p_{x_i} \right ) \mu_p.
\end{split}
\end{equation}
%
% \begin{equation}
% \label{E:MomentumMapGeneral}
%\begin{split}
%J(x_1,x_2,y_1,y_2,\dot x_1,\dot x_2, \dot y_1, \dot y_2)&= \left ( \sum_{i=1}^2m_i\frac{x_i\dot x_i+y_i\dot y_i}{y_i^2} \right ) \mu_h
%+\left ( \sum_{i=1}^2m_i\frac{\dot x_i(y_i^2-x_i^2-1)-2\dot y_ix_iy_i}{2y_i^2} \right ) \mu_e \\
%& \qquad + \left ( \sum_{i=1}^2\frac{ m_i\dot x_i}{y_i^2} \right ) \mu_p^+.
%\end{split}
%\end{equation}

The components of the above expression can be checked to be integrals of the equations of motion. Since
the symmetries associated to $\xi_h$ and $\xi_p$ 
are related to the homogeneity of space,  the components of $\mu_h$ and $\mu_p$ are natural generalizations of
linear momentum. On the other hand, since the symmetry defined by $\xi_e$ is a consequence of  isotropy of the space, the component of
$\mu_e$  is naturally  analogous to the classical   angular momentum. Our interpretation of the integrals of motion differs from that of \cite{FNDiacu}.

\begin{remark}
\label{R:Geom-Mech}
The system has the same number of conserved quantities as the euclidean two-body problem
arising from the action of the group of orientation preserving isometries. However notice that 
the subgroup representing translations is $\R^2$ in the euclidean case whereas for the hyperbolic
case it is $\Ht$. A fundamental difference between the two is that the latter is not abelian. As a consequence, 
 the symplectic
reduction of $T^*Q$ by $\Ht$ yields a reduced space whose generic dimension is 6 (the stabilizer of the momentum 
$G_\mu$ is generically trivial). This contrasts with the
euclidean case where the symplectic reduction  by the action of $\R^2$ has dimension 4 and the resulting reduced system
is equivalent
to the euclidean Kepler problem in the plane. The two extra dimensions
that are  lowered in the reduction of the euclidean case are interpreted as  
passing to the  center of mass coordinates, a procedure that is not
possible in the hyperbolic case.

Another difference is related to the fact that $\R^2$ is a normal subgroup of the euclidean group ${\rm SE}(2,\R)$
so one can perform reduction by stages \cite{RedStages} and further reduce the system by rotations. The ultimately
reduced system is the Hamiltonian system with one-degree of freedom  for the radius under the influence of the effective potential
of the Kepler problem.
 On the other hand, $\Ht$ is not a normal subgroup of ${\rm SL}(2,\R)$ and, as a consequence, 
the hypotheses for  reduction by stages \cite{RedStages} are not met.

\end{remark}

\section{Existence of relative equilibria}
\label{S:Existence}

In this section we will classify all relative equilibria of the problem. We will show that the only
relative equilibria arise as conjugation of the ones found in \cite{Diacu2}.

We begin with the following.
\begin{definition}
The {\em hyperbolic center of mass of two masses} $m_1$, $m_2$ in $\htt$
is the unique point in $\htt$ that satisfies the following properties:
\begin{enumerate}
\item It lies along the (unique) geodesic connecting $m_1$ and $m_2$.
\item If $d_1$ denotes its distance to $m_1$ and $d_2$ its distance to $m_2$, then
\begin{equation*}
m_1\sinh(2d_1) = m_2 \sinh(2d_2).
\end{equation*}
\end{enumerate}
\end{definition}
Note that in the above definition $d_1+d_2$ is the distance between $m_1$ and $m_2$.
It can be checked that the hyperbolic center of mass of two masses is well and uniquely defined.
Figure \ref{F:Hyp-Center-of-mass} illustrates our definition.
\begin{figure}[ht]
\centering
\includegraphics[width=9cm]{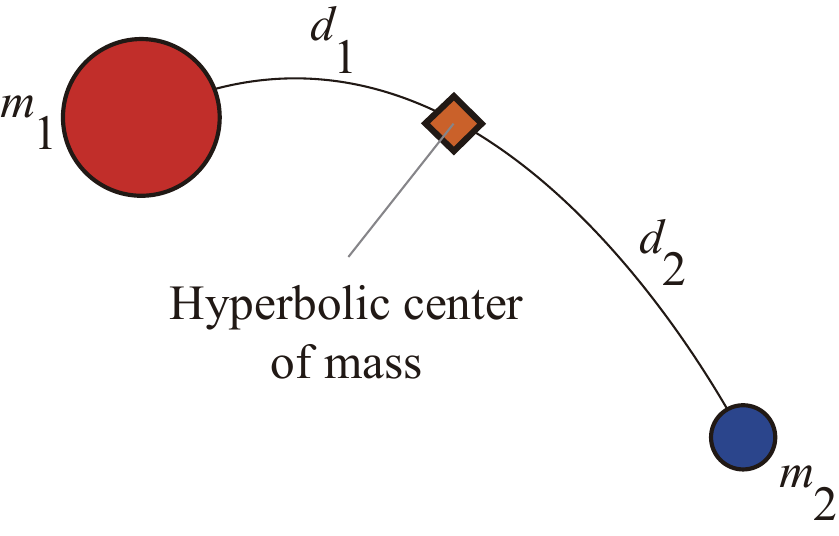}
\caption{\small{Schematic representation of the hyperbolic center of mass between particles $m_1$ and $m_2$.
The curve connecting the particles is a geodesic and the relation $m_1\sinh(2d_1) = m_2 \sinh(2d_2)$ is satisfied. }}\label{F:Hyp-Center-of-mass}
\end{figure}

Suppose that $X_1, X_2\in \htt$ denote the positions of the masses $m_1$ and $m_2$  and that $C\in \htt$ is their hyperbolic 
center of mass. It is clear that the hyperbolic center of mass of the 
masses $m_1$ and $m_2$ at the positions $g\cdot X_1, g\cdot X_2$ is $g\cdot C$ where $g$ denotes any
isometry of $\htt$.
%
%{\color{blue}
%It is clear that if $g \in {\rm SL}(2, \R)$, $X, X'$ are material points in $\htt$ and $C$ is the hyperbolic 
%center of mass of $X$ and $X'$ then $g C$ is the hyperbolic center of mass of $g X$ and $g X'$.}
%
%{\color{green}
% I prefer the first paragraph since we have not defined material points properly.
%Note that the property is also true for isometries that do not preserve orientation.
%}

Our definition of the hyperbolic center of mass is the one appropriate for our purposes of classifying the relative equilibria of the problem.\footnote{We are not claiming
any general  properties of the evolution of the hyperbolic center of mass under the dynamics, see Remark \ref{R:Center-of-mass}.}
 Interestingly, our definition differs from the position of the hyperbolic centroid defined in \cite{Ga}  by generalizing the euclidean ``lever rule". In this treatment, the hyperbolic centroid is a material point with a certain mass. For us, the hyperbolic center of mass is just an abstract point in $\htt$.

\begin{definition}
\label{D:CanonicalConfiguration}
Two masses $m_1, m_2$ in $\Ht$ are said to be in {\em canonical configuration} if they are respectively located at
$(\cos \theta_1, \sin \theta_1)$ and $(-\cos \theta_2, \sin \theta_2)$, where the angles $\theta_1, \theta_2$ satisfy $0<\theta_1, \theta_2<\frac{\pi}{2}$, and the following relation holds
\begin{equation}
\label{E:TwoPHypREcond}
\frac{m_1}{m_2}=\frac{\cos \theta_2 \sin^2\theta_1}{\sin^2 \theta_2 \cos \theta_1}.
\end{equation}
\end{definition}

The definition for a canonical configuration is illustrated in Figure \ref{F:Can-Conf}.
\begin{figure}[ht]
\centering
\includegraphics[width=9cm]{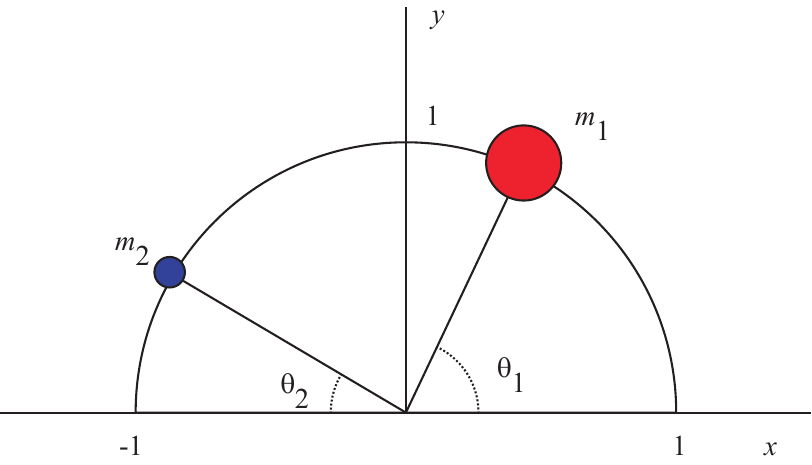}
\caption{\small{Two particles of masses $m_1$ and $m_2$ in canonical configuration.
The  relation \eqref{E:TwoPHypREcond} holds. }}\label{F:Can-Conf}
\end{figure}

Using (\ref{trigon-hyper}), we deduce the following characterization of two masses in canonical configuration.
\begin{lemma}\label{resul-conf-can}
Two masses in $\Ht$ which are located at $(\cos \theta_1, \sin \theta_1)$ and $(-\cos \theta_2, \sin \theta_2)$, with $0 < \theta_1, \theta_2 < \frac{\pi}{2}$, are in canonical configuration if and only if its center of mass is the point $(0, 1)$.
\end{lemma}
Now, we may prove the following result
\begin{proposition}
\label{P:Canonic-Conf}
Given two masses $m_1, m_2$ in $\Ht$ there exists an element of $\rm{SL}(2,\R)$ that maps them to the canonical configuration.
\end{proposition}
\begin{proof}
In view of Proposition \ref{P:Geod-unitcircle} there exists an isometry that maps the geodesic passing through $m_1$ and
$m_2$ onto the upper unit circle with the center of mass mapping onto $(0,1)$. If after this transformation $m_1$ lies in the first
quadrant we are done. Otherwise we apply the isometry defined by the matrix
\begin{equation*}
\left ( \begin{array}{cc} 0 & -1 \\ 1 & 0 \end{array} \right )\in \rm{SL}(2,\R).
\end{equation*}

\end{proof}

Relative equilibria are solutions of the equations of motion that are contained in a group orbit for the action of $\rm{SL}(2,\R)$ on
$T^*Q$. It is well known that a relative equilibrium is equivalent to a pair $(q,\xi)\in Q\times \frak{g}$, and in this case the solution
is given by $z(t)=\exp(t\xi)\cdot \langle \langle \xi_Q(q),\cdot \rangle \rangle$, where $\langle \langle \cdot, \cdot \rangle \rangle$ denotes the
Riemannian metric \eqref{E:Metric-on-Q} on $Q$. The element $\xi\in \frak{g}$ is called the velocity of the relative equilibrium. We will refer
interchangeably to both $z(t)$ and $p_q= \langle \langle \xi_Q(q),\cdot \rangle \rangle$ as the relative equilibrium.

If $z(t)$ is a relative equilibrium then so is $g\cdot z(t)$ for any $g\in G$ with the same stability properties as $z(t)$.
The velocity of $g\cdot z(t)$  is ${\rm Ad}_g\xi$.

Due to the equivariance of the momentum map, a relative equilibrium $p_q$ with velocity $\xi$ and $J(p_q)=\mu$ must satisfy
\begin{equation}
\label{E:Nec-Cond}
\rm{ad}_\xi^*\mu=0.
\end{equation}

Pairs $(q,\xi)$ corresponding to relative equilibria are characterized by the condition
\begin{equation}
\label{E:caractRE}
dV_\xi(q)=0,
\end{equation}
 where $V_\xi\in C^\infty(Q)$ is the
{\em augmented potential} defined by
\begin{equation*}
V_\xi(q)=V(q)-\frac{1}{2}\langle \I(q)\xi, \xi \rangle.
\end{equation*}
Here $ \I(q):\slt\to \slt ^*$ is the {\em locked inertia tensor} defined as
\begin{equation}
\label{E:Locked-In}
\langle \I(q)\xi, \eta \rangle =\langle \langle \xi_Q(q),\eta_Q(q)\rangle \rangle.
\end{equation}

Since relative equilibria come in group orbits, according to Proposition \ref{P:Canonic-Conf}, we restrict our  study of existence of relative equilibria assuming that the masses
are in canonical configuration.

Let $p_q$ be a relative equilibrium with velocity $\xi=E\xi_e+H\xi_h+P\xi_p$. One can readily check, using \eqref{E:J} that
\begin{equation}
\label{E:Mom-Trig}
J(p_q)=\frac{m_2(\cos\theta_2+\cos\theta_1)}{2\sin^2\theta_2\cos\theta_1}\left ( \begin{array}{cc} H &
(1-2\cos \theta_1\cos \theta_2)P+\cos\theta_1\cos\theta_2E \\
P-\cos\theta_1\cos\theta_2E & -H \end{array} \right )
\end{equation}
%
%
% be an arbitrary element in $\frak{se}(2)$. In view of \eqref{E:Infinit-Generator1}
%one obtains the following expression for the infinitesimal generator if $\xi$ at the canonical configuration $q\in Q$
%\begin{equation*}
%\xi_Q(q)=\left ( \begin{array}{c} H\cos \theta_1+P-E\cos^2\theta_1 \\ H\sin\theta_1-E\cos\theta_1\sin\theta_1\\ H\cos \theta_2+P-E\cos^2\theta_2 \\ H\sin\theta_2-E\cos\theta_2\sin\theta_2 \end{array} \right )\in T_qQ.
%\end{equation*}
%Using \eqref{E:MomentumMapGeneral} one finds that the momentum map evaluated at $(q,\xi_Q(q))\in T_qQ$ is  given by
%
where we have used \eqref{E:TwoPHypREcond}.

We start the classification of relative equilibria by applying condition \eqref{E:Nec-Cond} in order to isolate possible candidates.
%
%
%A necessary condition for the existence of a relative equilibria
%$(q,\xi)$ is that ${\rm ad}^*_\xi J(q,\xi_Q(q))=0$.
 In view of our identification %$\frak{se}(2)^*\cong \frak{se}(2)$ through the invariant pairing
 \eqref{E:Dual-identification} this is equivalent to the condition $[J(p_q), \xi]=0$ where $[\cdot , \cdot ]$ is the matrix
commutator. Up to
an unessential non-vanishing factor this commutator is
\begin{equation*}
\left ( \begin{array}{cc} P(E-P) & -H((1+\cos\theta_1\cos\theta_2)E-(1+2\cos\theta_1\cos\theta_2)P) \\
 -H((1+\cos\theta_1\cos\theta_2)E-P) & -P(E-P) \end{array} \right ).
\end{equation*}
There are only three non-trivial possibilities for the above matrix to be equal to zero:
\begin{enumerate}
\item[(a)] $E=P=0$, $H=\omega \neq 0$. The Lie algebra generator $\xi=\omega\xi_h$ is hyperbolic.
\item[(b)] $H=P=0$, $E=\omega \neq 0$. The Lie algebra generator $\xi=\omega\xi_e$ is elliptic.
\item[(c)]  $H=0$, $P=E=\omega\neq 0$, The Lie algebra generator $\xi=\omega(\xi_e+\xi_p)$ is hyperbolic.
\end{enumerate}

In particular notice that the above conditions do not include the possibility of having $H=E=0$ and $P\neq 0$. This, using Proposition \ref{ellip-hyper-para} and \ref{P:Canonic-Conf},
gives an alternative proof to
the following result given before in \cite{Diacu2}.
\begin{proposition}
There do not exist relative equilibria for the two-body problem in $\Ht$ with a parabolic Lie algebra generator.
\end{proposition}

We now use equation \eqref{E:caractRE} to investigate when
possibilities (a), (b), (c), are realized for the particular  potential
energy \eqref{E:Potential}.

Using \eqref{E:Infinit-Generator1} and \eqref{E:Locked-In} we obtain the following representation for the
locked inertia tensor referred to the bases $\{\xi_e,\xi_h,\xi_p\}$ and $\{\mu_e, \mu_h, \mu_p\}$.
\begin{equation*}
\I(x_1,y_1,x_2,y_2)=\left ( \begin{array}{ccc} I_{11} & I_{12} & I_{13} \\  I_{12} & I_{22} & I_{23} \\  I_{13} & I_{23} & I_{33} \end{array}
\right ),
\end{equation*}
with
\begin{equation*}
\begin{split}
I_{11}&=\frac{1}{4}\sum_{i=1}^2\frac{m_i((x_i^2+y_i^2+1)^2-4y_i^2)}{y_i^2}, \qquad I_{22}=\sum_{i=1}^2\frac{m_i(x_i^2+y_i^2)}{y_i^2},
\qquad  I_{33}=\sum_{i=1}^2\frac{m_i}{y_i^2}, \\
I_{12}&=-\frac{1}{2}\sum_{i=1}^2\frac{m_ix_i(1+x_i^2+y_i^2)}{y_i^2},\qquad I_{13}=-\frac{1}{2}\sum_{i=1}^2\frac{m_i(1+x_i^2-y_i^2)}{y_i^2},
\qquad I_{23}=\sum_{i=1}^2\frac{m_ix_i}{y_i^2}.
\end{split}
\end{equation*}

%
% For this matter recall that for
%any $q\in Q$  the {\em locked inertia tensor} (see e.g. \cite{}) is the operator $\I(q):\frak{se}(2)\to \frak{se}(2)^*$ defined by
%\begin{equation*}
%  \langle \I(q)\xi, \eta\rangle =\langle \langle \xi_Q(q), \eta_Q(q) \rangle \rangle_q.
%\end{equation*}
%
%Following \cite{Ma} for any $\xi\in \frak{se}(2)$ we define the {\em augmented potential} $V_\xi:Q\to \R$ by
%\begin{equation*}
%V_\xi(q)=V(q)-\frac{1}{2} \langle \I(q)\xi, \xi\rangle.
%\end{equation*}
%
%Since we are dealing with a simple mechanical system, then $(q,\xi)$ is a relative equilibrium if and only if $q$ is
%a critical point of $V_\xi$ (see e.g. \cite{Ma}). In our case we should verify this condition assuming that $q$ is a canonical
%configuration and that the element $\xi=E\xi_e+H\xi_h+P\xi_p\in \frak{se}(2)$ satisfies one of the conditions described in items (a), (b), (c), given above.

With the above formulae in hand, one obtains the following expressions  for the augmented potential $V_\xi$ in cases (a), (b) and (c).
\begin{equation*}
\begin{split}
V_{\omega\xi_h}(x_1,y_1,x_2,y_2)&=V(x_1,y_1,x_2,y_2)-\frac{\omega^2}{2}\sum_{i=1}^2\frac{m_i(x_i^2+y_i^2)}{y_i^2}\, , \\
V_{\omega\xi_e}(x_1,y_1,x_2,y_2)&=V(x_1,y_1,x_2,y_2)-\frac{\omega^2}{8}\sum_{i=1}^2\frac{m_i(x_i^2+(y_i-1)^2)(x_i^2+(y_i+1)^2)}{y_i^2}\, ,
\\
V_{\omega(\xi_e+\xi_p)}(x_1,y_1,x_2,y_2)&=V(x_1,y_1,x_2,y_2)-\frac{\omega^2}{8}\sum_{i=1}^2\frac{m_i((x_i-1)^2+y_i)^2)((x_i+1)^2+y_i^2)}{y_i^2}\, .
\end{split}
\end{equation*}

A direct calculation with MAPLE\texttrademark\,  that is too long to be included here,
%but that can be provided upon request,
shows that, in cases (a) and (b),  $q=(\cos\theta_1, \sin \theta_1, -\cos\theta_2, \sin\theta_2)$ is a critical point of $V_\xi$ if and only if
\begin{equation}
\label{E:omega-2Hyp-trig}
\omega^2=k\frac{\sin^2\theta_2 \sin^2 \theta_1}{(\cos \theta_2 + \cos \theta_1 )^2}\left ( \frac{\sin^2 \theta_2}{\cos \theta_2} m_1 \right ) = k\frac{\sin^2\theta_2 \sin^2 \theta_1}{(\cos \theta_2 + \cos \theta_1 )^2}\left ( \frac{\sin^2 \theta_1}{\cos \theta_1} m_2 \right ) .
\end{equation}
On the other hand, if $\xi$ satisfies condition (c), then a calculation with MAPLE\texttrademark\, shows that
  condition \eqref{E:caractRE}  implies
\begin{equation*}
m_2^2k \cos\theta_2\sin^2\theta_2\sin^4\theta_1(\cos \theta_1+\cos \theta_2)=0.
\end{equation*}
Considering that $k, m_2>0$, and the restrictions on $\theta_1, \theta_2$ coming from  Definition \ref{D:CanonicalConfiguration},
it follows that this condition can never hold. Therefore we have shown the following.

\begin{theorem}
\label{T:Exist-Rel-Eq}
For any values of $m_1, m_2>0,$ there exist exactly two families of relative equilibria for the two body problem in $\Ht$.
Representatives of  these families in a canonical configuration $q\in Q$ are $(q,\omega \xi_h)$ and $(q,\omega \xi_e)$ where
$\omega$ satisfies \eqref{E:omega-2Hyp-trig}.
\end{theorem}

The existence of a representative of the relative equilibrium $(q,\omega \xi_h)$ in the general case $m_1\neq m_2$  is proved in Theorem 7 of \cite{Diacu2}.
The condition of this theorem, given in equation (96)\footnote{We note that there is a typo in this equation and it 
should read $\frac{m_1}{m_2}=-\frac{\beta_1y_1}{\beta_2y_2}$ as it is easily deduced from the proof of the theorem.},   is readily seen to be equivalent 
to \eqref{E:TwoPHypREcond}.
The existence of the relative equilibrium
$(q,\omega \xi_e)$ for $m_1\neq m_2$ was established in Theorem 3 of \cite{Diacu2} using the Poincar\'e disk model. 
We have just shown that, up to conjugation,
these are the only ones. Following the terminology in  \cite{Diacu2}, and in accordance to the nature of the Lie algebra generator, we shall refer to these families  as  {\em hyperbolic} or {\em elliptic}, respectively.

Notice that the orbits in $\Ht$ generated by hyperbolic (respectively elliptic) generators are unbounded (respectively bounded).
The same is true for hyperbolic and elliptic relative equilibria in $Q$.

\begin{remark}
In chapter 4 of \cite{Duistermaat} a relative equilibrium is called {\em elliptic} if the one-parameter subgroup $\{\exp(\xi t) : t\in \R\}$
generated by its
velocity $\xi$ is dense in a torus subgroup of the corresponding symmetry group. Relative equilibria that are not elliptic are
called {\em runaway relative equilibria}. In our case,  the
uni-parametric subgroup
$\{\exp(\xi_et) : t\in \R\}$ of ${\rm SL}(2, \R)$ is closed and isomorphic to $\mathrm{S}^1$ (a one-dimensional torus) while the  uni-parametric group generated
by $\xi_h$ is not compact (and hence its closure cannot be contained in a torus). Thus, our terminology for elliptic relative equilibria is consistent with that
of \cite{Duistermaat}, while our hyperbolic relative equilibria correspond to runaway relative equilibria.  From the general theory developed in  \cite{Duistermaat} it follows
that our elliptic relative equilibria are periodic whereas our hyperbolic relative equilibria ``run out" of any compact set $K\subset T^*Q$ as $t\to \pm \infty$ (i.e. they are unbounded curves).
\end{remark}
%
%\begin{remark}
%An explanation for the existence of hyperbolic relative equilibria is that the gravitational attraction between
%the particles is exactly compensated by the tendency of geodesics to separate in the hyperbolic space.
%\end{remark}

\section{Nonlinear stability}
\label{S:Stability}

Let $z: \R \to T^*Q$ be one of the relative equilibria found in Section \ref{S:Existence} and $J(z(0)) = \mu$ the initial momentum associated with $z$. Then, as we know, $J(z(t)) = \mu$, for all $t$. So, if $z(t) = \exp(\xi t)\cdot z(0)$, with $\xi \in {\frak g}$ then, using the equivariance of the momentum map $J$, we deduce that $\mu = J(\exp(\xi t)\cdot z(0)) = \exp(\xi t) \cdot \mu$, which implies that $\exp(\xi t)$ belongs to $G_\mu$, where $G_\mu$ is the stabilizer of the
momentum $\mu$ under the coadjoint representation of ${\rm SL(2, \R)}$ on  $\slt ^*$. Thus, the relative equilibrium $z$ is contained in the $G_\mu$-orbit through $z(0)$.

In this section we discuss the $G_\mu$-stability of the relative equilibrium $z(t)$. The precise definition of $G_\mu$-stability (see \cite{Patrick}) is the following:

A relative equilibrium $z: \R \to T^*Q$ is said to be $G_\mu$-stable if for every $G_\mu$-invariant open neighborhood $U$ of $z(\R)$ there exists an open subset $W$ in $T^*Q$, such that $z(\R) \subseteq W \subseteq U$, and with the property that for every solution $z'$ %of the Hamilton equations 
with initial condition in $W$, we have that $z'(t) \in U$, for all $t$ .

Since we are dealing with a simple mechanical system, we use the {\em reduced Energy-Momentum method} (REM) of Simo et al \cite{Simo},
which studies the signature of the restricted Hessian $d^2_{p_q}h_\xi |_N$. Here, $h_\xi$ is the {\em augmented Hamiltonian} defined
as  (see \cite{Ma})
\begin{equation*}
h_\xi=h-\langle J(\cdot),\xi\rangle,
\end{equation*}
and $N$ is a complement to $\g_\mu\cdot p_q$ in  $T_zJ^{-1}(\mu)$, where $\mu=J(p_q)$.

The reduced Energy-Momentum provides a convenient block-diagonalization of this bilinear form. Let $p_q$ be
a relative equilibrium with velocity $\xi$ and momentum $\mu$. Then
\begin{equation}
\label{E:REM-matrix}
d^2_{p_q}h_\xi |_N\simeq \left ( \begin{array}{ccc} {\rm Ar}& 0 & 0 \\ 0 & (d_q^2V_\xi+ {\rm corr})|_{ V_{\rm int}}& 0\\ 0 & 0 &  K \end{array} \right )
\end{equation}
relative to a splitting
\begin{equation}
\label{E:REM-split}
N\simeq V_{\rm rig}\oplus V_{\rm int}\oplus V_{\rm int}^*.
\end{equation}
Here $K$ is a positive definite bilinear form on $V_{\rm int}^*$ and equations \eqref{E:REM-matrix} and \eqref{E:REM-split}
(and therefore, the reduced Energy-Momentum method) hold if the bilinear form ${\rm Ar}: V_{\rm rig}\times V_{\rm rig}\to \R$ is non-degenerate.

Assuming that $G_\mu$ %, the stabilizer of $\mu$ under the coadjoint representation,
 is compact,
a sufficient condition for $G_\mu$-stability %the nonlinear stability (in the sense of stability modulo $G_\mu$, see \cite{Patrick})
of the relative equilibrium $p_q$ is that $d^2_{p_q}h_\xi |_N$ be definite.
If the conditions of REM hold, then this is equivalent to both ${\rm Ar}$ and $d_q^2V_\xi+ {\rm corr} $ being positive definite.
On the other hand, an odd number of negative eigenvalues of $d^2_{p_q}h_\xi |_N$ implies spectral instability
and therefore, instability regardless of the compacity of $G_\mu$ \cite{Ma}.

Now we define the different spaces appearing in the splitting \eqref{E:REM-split}. First,
\begin{equation*}
V_{\rm rig}=\left \{ \lambda\in \g \, :\,  \langle \I(q)\lambda, \eta \rangle =0 \; \forall \eta \in \g_\mu \, \right \}
\end{equation*}
where $\g_\mu$ is the Lie algebra of $G_\mu$. Now let $(\g_\mu\cdot q)^\perp\subset T_qQ$ be the orthogonal
complement of $\g_\mu\cdot q$ with respect to the Riemannian metric \eqref{E:Metric-on-Q}, then
\begin{equation}
\label{E:DefVint}
V_{\rm int}=\left \{v\in (\g_\mu\cdot q)^\perp \, :\,  \I(q)^{-1}((D\I  \cdot v)(\xi))\in \g_\mu  \, \right \}.
\end{equation}

The block ${\rm Ar}$ is defined by
\begin{equation*}
{\rm Ar}(\lambda_1, \lambda_2)=\left \langle {\rm ad}^*_{\lambda_1}\mu \, , \, \I(q)^{-1}( {\rm ad}^*_{\lambda_2}\mu) +
{\rm ad}_{\lambda_2}(\I(q)^{-1}\mu) \right \rangle,
\end{equation*}
while
\begin{equation*}
{\rm corr}(v_1,v_2)= \left \langle (D\I  \cdot v_1)(\xi) \, , \, \I(q)^{-1}( (D\I  \cdot v_2)(\xi) ) \right \rangle .
\end{equation*}

For the sequel, it is useful to notice that in a canonical configuration, using \eqref{E:TwoPHypREcond}, the expression for the
locked inertia tensor  simplifies to
\begin{equation}
\label{E:Locked-Canonical}
\I( \theta_1,\theta_2)=\frac{m_2(\cos\theta_1+\cos\theta_2)}{\sin^2\theta_2}\left ( \begin{array}{ccc}
\cos\theta_2 & 0 &-\cos\theta_2 \\
0& \frac{1}{\cos\theta_1} & 0 \\
-\cos\theta_2 & 0 & \frac{1}{\cos\theta_1}
\end{array}\right )
\end{equation}

\subsection{Hyperbolic case}
\label{SS:Hyp-Stab}
We consider a hyperbolic relative equilibrium $(q,\xi)$ where $q=(\cos \theta_1, \sin \theta_1, -\cos \theta_2, \sin \theta_2)$,
$\xi = \omega \xi_h$, and where condition \eqref{E:omega-2Hyp-trig}
holds. Using \eqref{E:Mom-Trig} we obtain
\begin{equation*}
\mu=J(p_q)= \omega m_2\frac{(\cos\theta_1+\cos\theta_2)}{\sin^2\theta_2 \cos\theta_1}\mu_h.
\end{equation*}
where $\omega$ satisfies \eqref{E:omega-2Hyp-trig}

In view of (\ref{structure-constants}), \eqref{E:dualbasis} and \eqref{E:adstar}
we obtain that $\g_\mu$ is spanned by $\xi_h$.

Using \eqref{E:Locked-Canonical} the space $V_{\rm rig}$ is spanned by
\begin{equation*}
\zeta_1=\xi_e, \qquad \zeta_2=\xi_p.
\end{equation*}
With respect to this basis, the bilinear form ${\rm Ar}$ is represented by the matrix
\begin{equation*}
{\rm Ar}=\frac{m_2\omega^2(\cos\theta_1+\cos\theta_2)\cos\theta_2}{\sin^2\theta_2(1-\cos\theta_1\cos\theta_2)}\left ( \begin{array}{cc} 1& -1\\
-1  & \frac{1}{\cos^2\theta_1\cos^2\theta_2} \end{array}\right ).
\end{equation*}

Since ${\rm Ar}_{11}$ and $\det ({\rm Ar})$ are both positive, then ${\rm Ar}$ is positive definite and therefore
REM is applicable since \eqref{E:REM-matrix} and   \eqref{E:REM-split} hold.

We now compute the internal space $V_{\rm int}$. Since  $\g_\mu$ is spanned by $\xi_h$, in view of \eqref{E:Infinit-Generator1} we obtain
\begin{equation*}
\g_\mu\cdot q={\rm span}\left \{ \left ( \cos \theta_1 \, ,\, \sin\theta_1 \, ,\, -\cos\theta_2 \, ,\, \sin \theta_2 \right ) \right  \}.
\end{equation*}
From this
\begin{equation*}
(\g_\mu\cdot q)^\perp=\left \{ \left (-a\frac{\sin\theta_1}{\cos\theta_1}+b-c\frac{\sin\theta_2}{\cos\theta_2}, a, b ,
c \right )\, :\, a,b,c\in \R \, \right \}.
\end{equation*}

If we now use \eqref{E:DefVint} we find that the space $V_{\rm int}$ is one-dimensional. After choosing a particular
generator
\begin{equation*}
v=\left ( -\sin^2\theta_1\cos\theta_1(\cos^2\theta_2+1) , \sin\theta_1\cos^2\theta_1(\cos^2\theta_2+1), \sin^2\theta_2\cos\theta_2(\cos^2\theta_1+1) , \sin\theta_2\cos^2\theta_2(\cos^2\theta_1+1) \right ),
\end{equation*}
 we obtain
 \begin{equation}
 \label{E:crucial-expression-hyper}
(d_q^2V_\xi+ {\rm corr})|_{ V_{\rm int}}(v,v)=-\frac{km_2^2\cos\theta_2\sin^4\theta_1(\cos\theta_1\cos\theta_2+1)(\cos^2\theta_1+\sin^2\theta_1\cos^2\theta_2+3)}{(\cos\theta_1+\cos\theta_2)\cos(\theta_1)},
\end{equation}
which is negative.

Therefore, since $d^2_{p_q}h_\xi |_N$ has signature $(+,+,-,+)$, then every hyperbolic relative equilibrium is unstable.

\subsection{Elliptic case}

We consider an elliptic relative equilibrium $(q,\xi)$ where $q=(\cos \theta_1, \sin \theta_1, -\cos \theta_2, \sin \theta_2)$,
$\xi = \omega \xi_e$, and where condition \eqref{E:TwoPHypREcond}
holds. Using \eqref{E:J} we obtain
\begin{equation*}
\mu=J(p_q)=m_2\frac{\cos \theta_2(\cos\theta_1+\cos \theta_2)\omega}{\sin^2\theta_2}(\mu_e-\mu_p).
\end{equation*}
In view of \eqref{E:dualbasis} we obtain that
\begin{equation}
\label{E:Mom-Elliptic}
\mu=-m_2\frac{\cos \theta_2(\cos\theta_1+\cos \theta_2)\omega}{\sin^2\theta_2}\xi_e.
\end{equation}
From this is clear that $G_\mu\simeq {\rm S}^1$ with Lie algebra $\g_\mu={\rm span} \{ \xi_e \}$. In particular, $G_\mu$ is compact.

We now compute $V_{\rm rig}$. From expression  \eqref{E:Locked-Canonical} we find that it is spanned by
\begin{equation*}
\zeta_1=\xi_h, \qquad \zeta_2=\xi_e+
\xi_p.
\end{equation*}
With respect to this basis, the bilinear form ${\rm Ar}$ is represented by the matrix
\begin{equation*}
{\rm Ar}=\frac{m_2\omega^2(\cos\theta_1+\cos\theta_2)\cos\theta_2}{\sin^2\theta_2}
\left ( \begin{array}{cc} \frac{1}{1-\cos\theta_1\cos\theta_2} & 0 \\ 0 & 1+\cos\theta_1\cos\theta_2 \end{array} \right ).
\end{equation*}
Since ${\rm Ar}_{11}$ and  ${\rm Ar}_{22}$  are both positive, then ${\rm Ar}$ is positive definite and therefore
REM is applicable since \eqref{E:REM-matrix} and   \eqref{E:REM-split} hold.

We now compute the internal space $V_{\rm int}$. We start by obtaining $\g_\mu\cdot q$. From the above expression for $\g_\mu$,
and \eqref{E:Infinit-Generator1} we obtain
\begin{equation*}
\g_\mu\cdot q={\rm span}\left \{ \left ( -\cos^2\theta_1 \, ,\, -\cos\theta_1\sin\theta_1 \, ,\, -\cos^2\theta_2  \, ,\, \cos \theta_2\sin \theta_2 \right ) \right  \}.
\end{equation*}
From this
\begin{equation*}
(\g_\mu\cdot q)^\perp=\left \{ \left (a, \frac{c\sin\theta_2-a\cos\theta_1-b\cos\theta_2}{\sin\theta_1}, b,  c \right )\, :\, a,b,c\in \R \, \right \}.
\end{equation*}

If we now use \eqref{E:DefVint} we find that the space $V_{\rm int}$ is one-dimensional.
 After choosing a particular
generator
\begin{equation*}
w=\left ( -\sin^2\theta_1\cos\theta_1(\cos^2\theta_2+1) , \sin\theta_1\cos^2\theta_1(\cos^2\theta_2+1), \sin^2\theta_2\cos\theta_2(\cos^2\theta_1+1) , \sin\theta_2\cos^2\theta_2(\cos^2\theta_1+1) \right ),
\end{equation*}
we obtain
 \begin{equation}
 \label{E:crucial-expression-ellip}
(d_q^2V_\xi+ {\rm corr})|_{ V_{\rm int}}(w,w)=\frac{m_2^2k v(1-u^2)^2(1+uv)F(u,v)}{u(u+v)},
\end{equation}
with
\begin{equation}
\label{E:Fuv}
F(u,v)=1-3u^2v^2-u^2-v^2,
\end{equation}
where $u=\cos \theta_1$ and $v=\cos \theta_2$. It is readily seen that the sign of \eqref{E:crucial-expression-ellip}
is the sign of $F(u,v)$.

Notice that $u$ and $v$ are related by \eqref{E:TwoPHypREcond}. We can express
\begin{equation}
\label{E:vofu}
v(u)=\frac{u^2-1+\sqrt{(u^2-1)^2+4c^2u^2}}{2cu},
\end{equation}
where $c=\frac{m_1}{m_2}$. It is readily seen that
\begin{equation*}
\lim_{u\to0}v(u)=0, \qquad \lim_{u\to1}v(u)=1, \qquad v'(u)>0.
\end{equation*}
Therefore, $(d_q^2V_\xi+ {\rm corr})|_{ V_{\rm int}}(w,w)$ is a positive multiple of $f(u)=F(u,v(u))$.
It is easy to see that
\begin{equation*}
\lim_{u\to0}f(u)=1, \qquad \lim_{u\to1}f(u)=-4.
\end{equation*}
Using the chain rule we find that $f$ is decreasing and therefore  $f(u)$ changes sign at a unique value $u_0\in (0,1)$.
Therefore, for $0<u<u_0$ the signature of $d^2_{p_q}h_\xi |_N$ is $(+,+,+,+)$ and the elliptic relative equilibrium
is nonlinearly stable, and for  $u_0<u<1$ the signature of $d^2_{p_q}h_\xi |_N$ is $(+,+,-,+)$ and the elliptic relative equilibrium
is unstable.

Substitution of \eqref{E:vofu} into  \eqref{E:Fuv} shows, after some algebraic manipulations,
that $u_0$ is characterized as the unique zero between $0$ and $1$ of the polynomial
\begin{equation*}
p(x)=3x^8+(16c^2-8)x^6+6x^4-1=(3x^2+1)(x^2-1)^3+16c^2x^6.
\end{equation*}
More precisely, the elliptic relative equilibrium is stable if $p(u)<0$ and unstable if $p(u)>0$. Recalling that  $u=\cos\theta_1$, 
 the condition for stability can therefore be written as
\begin{equation}
\label{E:StabCond}
\frac{m_1}{m_2}<\frac{\sin^3\theta_1\sqrt{3\cos^2\theta_1+1}}{4\cos^3\theta_1}.
\end{equation}
This formula will be useful to give an intrinsic description of the stability conditions of elliptic relative equilibria
in section \ref{S:Intrinsic}.

In view of \eqref{E:vofu}, the momentum \eqref{E:Mom-Elliptic} can be written as a function of $u$.
The change in stability is related to a critical value of the momentum $\mu(u)$ as the following proposition shows. This behavior of the stability of a parametrized family of relative equilibria with respect to a critical point of the norm of the momentum seems to be more than a mere coincidence since it has been observed in a similar context in other systems. See for instance Remark 3.13 in \cite{Teix-Ro}.

\begin{proposition}
Let $||\cdot ||$ be any vector space norm in $\slt ^*$; then $||\mu(u)||$ is increasing for  $0<u<u_0$ and decreasing for $u_0<u<1$,
where $u_0$ is the unique zero of $f(u)=F(u,v(u))$ in the interval $(0,1)$.
\end{proposition}
\begin{proof}
Substitution of \eqref{E:omega-2Hyp-trig} into \eqref{E:Mom-Elliptic} (using (\ref{E:TwoPHypREcond}))
 yields
\begin{equation*}
\mu=\mp\sqrt{km_2}m_1\sqrt{\cos\theta_1(1-\cos^2\theta_2)}\xi_e=\mp\sqrt{km_2}m_1\sqrt{u(1-v^2)}\xi_e.
\end{equation*}
Therefore,
\begin{equation*}
||\mu(u)||=\lambda\sqrt{u(1-v^2)},
\end{equation*}
where the positive number $\lambda=\sqrt{km_2}m_1||\xi_e||$. We have
\begin{equation*}
\frac{d||\mu(u)||}{du}=\lambda\frac{1-v^2-2uvv'}{2\sqrt{u(1-v^2)}}.
\end{equation*}
On the other hand, the relation  \eqref{E:TwoPHypREcond} is written in terms of $u$ and $v$ as
\begin{equation*}
c=\left ( \frac{1-u^2}{u}\right ) \left (\frac{v}{1-v^2} \right ).
\end{equation*}
Implicit differentiation with respect to $u$ yields
\begin{equation*}
v'=\frac{(u^2+1)(1-v^2)v}{(1-u^2)(1+v^2)u}.
\end{equation*}
Hence
\begin{equation*}
\frac{d||\mu(u)||}{du}=\lambda\frac{(1-v^2)(1-3u^2v^2-u^2-v^2)}{2\sqrt{u(1-v^2)}(1-u^2)(1+v^2)}=\lambda\frac{(1-v^2)F(u,v)}{2\sqrt{u(1-v^2)}(1-u^2)(1+v^2)},
\end{equation*}
that has the sign of $F(u,v)$.

\end{proof}

Two physical implications of  $G_\mu$-stability % for the $G_\mu$-stable  elliptic relative equilibria that we found 
may be given using a general result for symmetric Hamiltonian systems and a standard result for the $2$-body problem in an arbitrary Riemannian manifold (see the Appendix). In fact, let $z: \R \to T^*Q$ be a $G_\mu$-stable  relative equilibrium. Then, we have:
\begin{enumerate}
\item
All the solutions sufficiently close to $z$ are bounded. 
\item
If $r_0 > 0$ is the constant hyperbolic distance between the two particles along $z$, it follows that the hyperbolic distance between the two particles remains arbitrarily close to $r_0$ along all the solutions that are sufficiently close to $z$.
\end{enumerate}
 Properties (i) and (ii) follow using Propositions \ref{bounded} and \ref{constant-distance} (in the Appendix), respectively. For (i) to hold it is essential that $G_\mu$ is compact.

%\begin{remark}
%\textcolor{Blue}{As the discussion in the Appendix shows, the above physical implication is a consequence of   $G$-stability. This notion of stability is weaker  
  %than $G_\mu$-stability.
%The precise physical implications of $G_\mu$-stability in our problem are more intricate and we do not attempt to describe them here. }
%\end{remark}

\section{Intrinsic description}
\label{S:Intrinsic}

We  now give an intrinsic description of the relative equilibria of the problem  involving only  Riemannian data and without referring to a
canonical configuration. Therefore, this description will be valid for any  model of $\htt$.
% once we represent its isometry group  as ${\rm PSL}(2,\R)$.

\begin{theorem}
\label{T:HypN2}
Consider two masses $m_1$ and $m_2$ in the abstract hyperbolic two-dimensional space $\htt$. A hyperbolic relative
equilibrium corresponds to the solution characterized in the following way.
 Let $d_i$,  $i=1,2$ denote the distance between the mass $m_i$ and  the hyperbolic center of mass of $m_1$ and $m_2$, $d:=d_1+d_2$ the
distance between the masses,
and let $\omega^2$ satisfy:
\begin{equation*}
\omega^2=\frac{2km_1}{\sinh^2(d)\sinh(2d_2)}=\frac{2km_2}{\sinh^2(d)\sinh(2d_1)}.
\end{equation*}

Then, at every instant during the motion,
\begin{enumerate}
\item $d_i$ is constant,
\item  the velocity vectors  $v_1$ and $v_2$ are  perpendicular
to the unique geodesic that contains $m_1$ and $m_2$ and are equally oriented,
\item the hyperbolic norm of $v_i$ is constant and equal to  $|\omega| \cosh(d_i)$,
\item  the hyperbolic center of mass moves along the geodesic that is perpendicular to the geodesic containing $m_1, m_2$.
Its velocity vector $v$ has constant hyperbolic norm equal to $|\omega|$.
\end{enumerate}
\end{theorem}

\begin{proof}
It is straightforward to see that (i) holds since the action is by isometries. To prove the other statements we work
 on the $\Ht$ model. By Proposition \ref{P:Canonic-Conf} it is sufficient to show that the statements hold for the
 relative equilibrium in canonical configuration. Using \eqref{E:Hyp-Par} it is seen that such relative equilibrium is given by
 \begin{equation*}
(x_1(t),y_1(t))=e^{t\omega}(\cos\theta_1,\sin \theta_1), \qquad (x_2(t),y_2(t))=e^{t\omega}(-\cos\theta_2,\sin \theta_2),
\end{equation*}
where $\omega$ is given by \eqref{E:omega-2Hyp-trig}.

At the time $t$, the masses lie along the geodesic $\ell_t$ that consists of the upper
half of the circle of radius $e^{\omega t}$  centered at $(0,0)$. The velocity vector of each particle points along the radial direction
and is clearly perpendicular to this geodesic (we are using that the $\Ht$ model  is conformal). Moreover, since $e^{t\omega} \sin \theta_1$ and $e^{t\omega} \sin \theta_2$ have the same sign, the velocity vectors
are seen to have the same orientation so we have proved (ii).

Next, the vector $v_1=\omega e^{t\omega}(\cos\theta_1,\sin \theta_1)$ and its hyperbolic norm is given by
\begin{equation*}
\sqrt{\frac{(\omega e^{t\omega}\cos\theta_1)^2+( \omega e^{t\omega}\sin\theta_1)^2)}{(y_1(t))^2}}=\frac{|\omega|}{\sin\theta_1}.
\end{equation*}
Substituting  \ref{trigon-hyper} into \eqref{E:omega-2Hyp-trig}
yields the expression for $\omega$ given in the statement of the theorem. Equation  \eqref{trigon-hyper}
also implies that $\frac{1}{\sin\theta_1}=\cosh(d_1)$. So we have shown that (iii) holds for $v_1$. An analogous  calculation
shows that it holds for $v_2$.

Finally, note that the hyperbolic center of mass is located at $e^{t\omega}(0,1) = (0,e^{t\omega})$ at time $t$. It therefore traverses  a vertical
geodesic that is always perpendicular to $\ell_t$. The hyperbolic norm of its velocity vector is readily computed to be $|\omega|$.

\end{proof}

A schematic description of the initial conditions that give rise to a hyperbolic relative equilibrium is given in Figure \ref{F:Hyp-RE}.

\begin{figure}[ht]

\centering
\includegraphics[width=12cm]{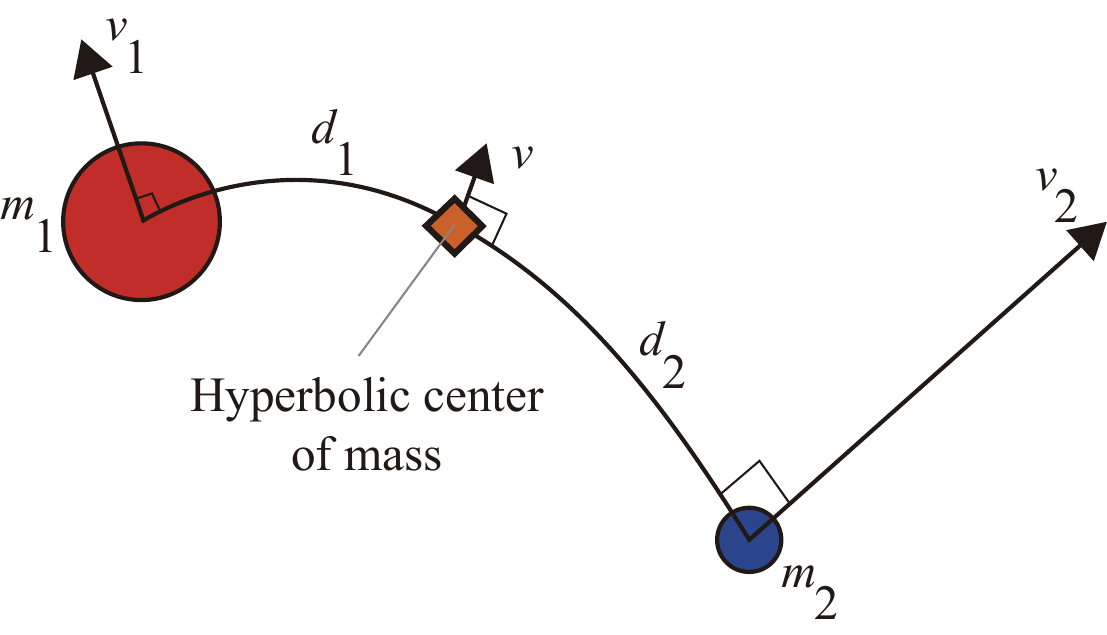}
\caption{Schematic representation of the initial conditions for a hyperbolic relative equilibrium. The vectors $v_1, v_2$ are perpendicular to the geodesic passing through $m_1, m_2$. Their respective hyperbolic norms are $|\omega|\cosh(d_i)$ where $\omega$ satisfies \eqref{E:OmEllip}. During the motion, the hyperbolic
center of mass moves along a geodesic that is perpendicular to the one connecting $m_1$ and $m_2$ and its velocity vector $v$ has
constant hyperbolic norm $|\omega|$.}\label{F:Hyp-RE}
\end{figure}

The above theorem has a nice geometrical interpretation. Consider two particles in $\htt$ whose velocities are
perpendicular to the geodesic that contains them and  point in the same direction.  There are two competing
effects that act on the particles.
 On the one hand, due to their inertia, they tend to follow  geodesic curves. These curves spread apart since we
are in hyperbolic space. On the other hand, there is the gravitational effect that pulls the particles together. The theorem tells
us that there exist unique, constant values of the particles' speeds that exactly balance the above effects, and the distance between
the particles remains constant throughout the motion. This type of motion is impossible in euclidean space since in this space parallel  lines
do not spread apart.

Our discussion in Section \ref{SS:Hyp-Stab} shows that all of these relative equilibria are unstable.

The intrinsic description of elliptic relative equilibria is given by the following.

\begin{theorem}
\label{T:EllN2}
Consider two masses $m_1$ and $m_2$ in the abstract hyperbolic two-dimensional space $\htt$. An elliptic relative
equilibrium corresponds to the solution characterized in the following way.

 Let $d_i$,  $i=1,2$ denote the distance between the mass $m_i$ and  the hyperbolic center of mass of $m_1$ and $m_2$, $d:=d_1+d_2$ the
distance between the masses,
and let $\omega^2$ satisfy:
\begin{equation}
\label{E:OmEllip}
\omega^2=\frac{2km_1}{\sinh^2(d)\sinh(2d_2)}=\frac{2km_2}{\sinh^2(d)\sinh(2d_1)}.
\end{equation}

Then, throughout the motion,
\begin{enumerate}
\item $d_i$ is constant,
\item  the velocity vectors  $v_1$ and $v_2$ are  perpendicular
to the unique geodesic that contains $m_1$ and $m_2$ and have opposite orientations,
\item the hyperbolic norm of $v_i$ is constant and equal to  $|\omega| \sinh(d_i)$,
\item  the hyperbolic center of mass is fixed during the motion,
\item the motion is periodic with period $\frac{2\pi}{|\omega|}$.
\end{enumerate}
\end{theorem}
\begin{proof}
We proceed in analogy with the proof of Theorem \ref{T:HypN2}
and work with a canonical configuration in the $\Ht$-model. According to
\eqref{E:Ellip}, such relative equilibrium is given by
 \begin{equation*}
(x_1(t),y_1(t))=\frac{(\cos\theta_1\cos (t\omega),\sin \theta_1)}{1+\cos\theta_1\sin (t\omega)}, \qquad
(x_2(t),y_2(t))=\frac{(-\cos\theta_2\cos(t\omega),\sin \theta_2)}{1-\cos\theta_2\sin (t\omega)},
\end{equation*}
where $\omega$ is given by \eqref{E:omega-2Hyp-trig}. Recall that substitution of \eqref{trigon-hyper} into \eqref{E:omega-2Hyp-trig}
yields the expression for $\omega$ given in the statement of the theorem. The above formulae prove (v).

As before, statement (i) follows since the action is by isometries.
Using \eqref{E:Ellip} with $x=0, y=1,$ shows that
the hyperbolic center of mass is fixed during the motion so we have shown (iv).

 Next, the geodesic that contains $m_1$ and $m_2$ at time $t$ is either the semi-circle
\begin{equation*}
\ell_t=\left \{ (x,y)\in \Ht \, :\, \left (x+\tan(t\omega) \right )^2+y^2={\rm sec}^2(t\omega) \right \} \qquad {\rm if } \qquad t\neq \frac{(k+\frac{1}{2})\pi}{\omega}, \; k\in \Z,
\end{equation*}
or the vertical line
\begin{equation*}
\ell_t=\left \{ (x,y)\in \Ht \, :\, x=0 \right \} \qquad {\rm if } \qquad t= \frac{(k+\frac{1}{2})\pi}{\omega}, \; k\in \Z.
\end{equation*}
In either case, $\ell_t$ is conveniently parametrized by hyperbolic arc-length as
\begin{equation}
\label{E:Param-geod}
x_g(s)=\frac{\tanh(s)\cos(t\omega)}{1+\tanh(s)\sin(t\omega)}, \qquad  y_g(s)=\frac{\sech(s)}{1+\tanh(s)\sin(t\omega)}, \qquad s\in \R.
\end{equation}
Using \eqref{trigon-hyper} one checks that such parametrization satisfies
\begin{equation*}
(x_g(d_1),y_g(d_1))=(x_1(t),y_1(t)), \qquad (x_g(-d_2),y_g(-d_2))=(x_2(t),y_2(t)).
\end{equation*}
One computes
\begin{equation*}
x_g'(s)=\frac{\sech^2(s)\cos(t\omega)}{(1+\tanh(s)\sin(t\omega))^2}, \qquad
y_g'(s)=-\frac{\tanh(s) \sech(s)+\sech(s)\sin(t\omega)}{(1+\tanh(s)\sin(t\omega))^2},
\end{equation*}
 where $'=\frac{d}{ds}$. Using \eqref{trigon-hyper} it is seen that the velocity vectors $v_i=(\dot x_i(t), \dot y_i(t))$, $i=1,2$ can be written as
 \begin{equation*}
v_1=\omega \frac{\cos\theta_1}{\sin\theta_1}\left (y_g'(d_1),-x_g'(d_1) \right ), \qquad v_2=-\omega \frac{\cos\theta_2}{\sin\theta_2}
\left (y_g'(-d_2),-x_g'(-d_2)\right ).
\end{equation*}
Since $\Ht$ is conformal, this shows that both vectors are perpendicular to $\ell_t$ and are oppositely oriented (note that $x'_g(d_1)$ and $x'_g(-d_2)$ have the same sign). This proves (ii).
Moreover, from the above expression we have that the hyperbolic norm of $v_1$ is
\begin{equation*}
\left | \omega \frac{\cos\theta_1}{\sin\theta_1} \right |\sqrt{\frac{y_g'(d_1)^2+x_g'(d_1)^2}{y_g(d_1)^2}}=| \omega |\frac{\cos\theta_1}{\sin\theta_1}=|\omega|\sinh(d_1)
\end{equation*}
where we have used \eqref{trigon-hyper} and that the parametrization \eqref{E:Param-geod} is by hyperbolic arc-length.
Analogously, the hyperbolic norm of $v_2$ is $|\omega|\sinh(d_2)$ and we have shown (iii).
\end{proof}

A schematic description of the initial conditions that give rise to an elliptic relative equilibrium is given in Figure \ref{F:Ellip-RE}.

\begin{figure}[ht]

\centering
\includegraphics[width=12cm]{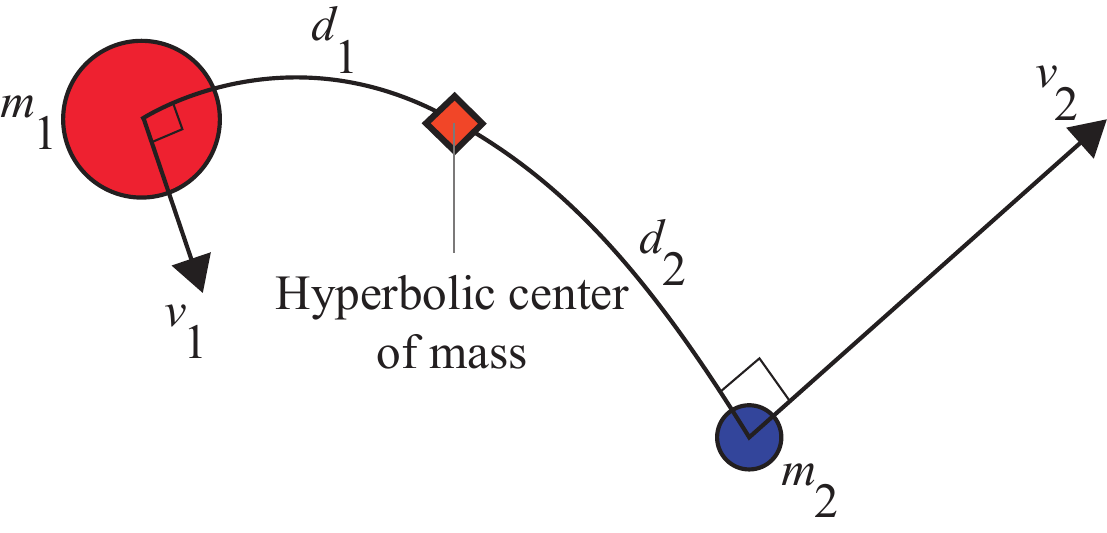}
\caption{Schematic representation of the initial conditions for an elliptic relative equilibrium. The vectors $v_1, v_2$ are perpendicular to the geodesic passing through $m_1, m_2$. Their respective hyperbolic norms are $|\omega|\sinh(d_i)$ where $\omega$ satisfies \eqref{E:OmEllip}. During the motion, the hyperbolic
center of mass is at rest.}\label{F:Ellip-RE}
\end{figure}

The stability results for elliptic relative equilibria can also be formulated in an intrinsic way. For this matter recall that the
 stability condition can be expressed by  equation \eqref{E:StabCond}. Using  \eqref{trigon-hyper} we conclude
that the elliptic relative equilibrium described in Theorem \ref{T:EllN2} is $G_\mu$-stable if
\begin{equation*}
\frac{m_1}{m_2}< \frac{\sqrt{3\tanh^2(d_1)+1}}{4\sinh^3(d_1)},
\end{equation*}
and it is unstable if the reverse inequality holds.

 Figure \ref{F:Stability-diagram} illustrates the values of $d_1$ for which the solution is stable for a given ratio $\frac{m_1}{m_2}$. In essence, the stability is lost when the masses are far away.
\begin{figure}[ht]

\centering
\includegraphics[width=6cm]{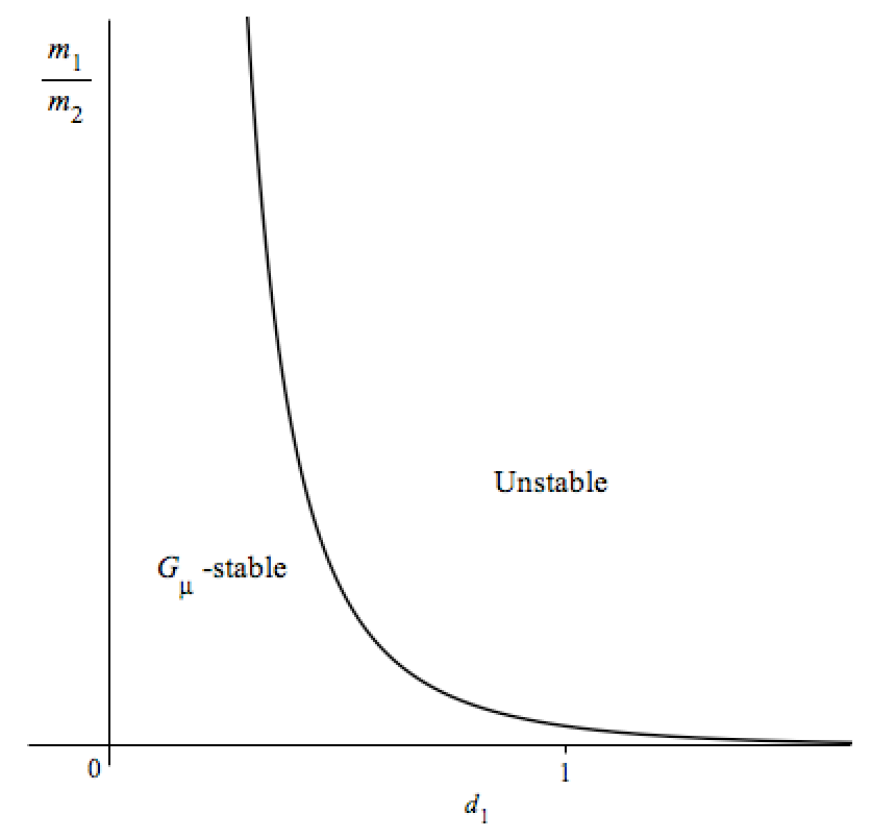}\caption{Stability properties of elliptic relative equilibria as a function of $d_1$ and the mass ratio  $\frac{m_1}{m_2}$.}  \label{F:Stability-diagram}
\end{figure}

\begin{remark}
\label{R:Center-of-mass}
It is remarkable that the hyperbolic center of mass in both types of relative equilibria has a behavior that is similar to its euclidean counterpart. Namely, it follows a geodesic at constant speed or it
stays fixed during the motion. As numerical calculations show, this kind of behavior should not be expected for generic solutions of the problem. This is usually referred to as the ``absence of the integral of center of mass"
(see \cite{FNDiacu}).

\end{remark}

%
%
%
%\begin{theorem}
%\label{T:EllN2}
%Consider two masses $m_1$ and $m_2$ in hyperbolic two-dimensional space.
%\begin{enumerate}
%\item There exists an elliptic relative equilibrium
%for such configuration.
%
%\item Let $d_i$,  $i=1,2$ denote the distance between the mass $m_i$ and  the hyperbolic center of mass of $m_1$ and $m_2$, $d:=d_1+d_2$ the
%distance between the masses,
%and let $\omega^2$ satisfy:
%\begin{equation*}
%\omega^2=\frac{2m_1}{\sinh^2(d)\sinh(2d_2)}=\frac{2m_2}{\sinh^2(d)\sinh(2d_1)}.
%\end{equation*}
%
%Then, at every instant during the motion, $d_i$ is constant and the mass $m_i$   moves at uniform speed $v_i=|\omega| \sinh(d_i)$  and perpendicular
%to the unique geodesic that contains $m_1$ and $m_2$ and around the hyperbolic circle of radius $d_i$ and center in the hyperbolic center of mass. Moreover,  the hyperbolic center of mass stays fixed throughout the motion.
%\end{enumerate}
%\end{theorem}

\renewcommand{\thesection}{\sc{Appendix}}

\section{}
\label{A:almost-constant-distance}

\renewcommand{\thesection}{A}

In this Appendix, we will prove some general results on $G_\mu$-stable relative equilibria for symmetric Hamiltonian systems.  Although these results seem to be quite
standard, we include them for the sake of completeness.

Let $(P, \Omega)$ be a symplectic manifold with symplectic $2$-form $\Omega$, $G$ a Lie group which acts symplectically on $P$, $J: P \to {\frak g}^*$ an equivariant momentum map and $H: P \to \R$ a $G$-invariant Hamiltonian function.

Suppose that $z: \R \to P$ is a relative equilibrium,
\[
z(t) = \exp(\xi t) \cdot z(0), \; \; \mbox{ with } \xi \in {\frak g},
\]
and that $J(z(0)) = \mu$. Using Noether Theorem, the equivariance of the momentum map $J$ and proceeding as at the beginning of Section \ref{S:Stability}, we deduce that $z(\R) \subseteq G_\mu \cdot z(0)$, that is, $z(\R)$ is contained in the $G_\mu$-orbit by $z(0)$.
 
Next, assume that $z$ is a $G_\mu$-stable relative equilibrium. This means that for every $G_\mu$-invariant open subset $U$ of $P$, with $z(\R) \subseteq U$, there exists an open subset $W$ of $P$ such that $z(\R) \subseteq W \subseteq U$ and for every solution $z'$ of the Hamilton equations with initial condition in $W$, we have that $z'(t) \in U$, for all $t$ (see \cite{Patrick}).

Now, we will consider two special cases:

{\bf 1. The isotropy group $G_\mu$ is compact.}
In this case, we will prove that solutions sufficiently close to $z$ are bounded. In fact, we will see that the following result holds.
\begin{proposition} \label{bounded}
If $z: \R \to P$ is a $G_\mu$-stable relative equilibrium and $G_\mu$ is compact then there exists an open subset $W$ of $T^*Q$, such that $z(\R) \subseteq W$, and for every solution $z'$ of the Hamilton equations with initial condition in $W$ we have that $z'$ is bounded.
\end{proposition}
\begin{proof}
Let $U_0$ be an open subset of $P$, with $z(0) \in U_0$, and $K$ a compact subset such that $U_0 \subset K$. Consider the open subset $U$ of $P$ given by
\[
U = G_\mu \cdot U_0.
\]
It is clear that $U \subseteq G_\mu \cdot K$ and $G_\mu \cdot K$ is compact. Thus, $U$ is bounded. Moreover, $U$ is $G_\mu$-invariant and $z(\R) \subseteq U$. Therefore, there exists an open subset $W$ in $P$, such that for every solution $z'$  with initial condition in $W$, we have that $z'(t) \in U$, for all $t$. Now, using that $U$ is bounded, we conclude that $z'$ also is a bounded curve.
\end{proof}
{\bf 2. The Hamiltonian system is the $2$-body problem in an arbitrary Riemannian manifold.}
In this case, $P$ is the cotangent bundle of a manifold $Q$ endowed with the canonical symplectic structure, where 
\[
Q = (M \times M) \setminus \Delta,
\]
with $M$ an arbitrary Riemannian manifold and $\Delta$ the diagonal in $M \times M$.

Moreover, we will suppose that $G$ is a Lie subgroup of the isometry Lie group of $M$ and the action of $G$ on $T^*Q$ is the cotangent lift of the diagonal action of $G$ on $Q$.

In addition, the Hamiltonian function $H: T^*Q \to \R$ on $T^*Q$ is given by
\[
H(p_1, p_2) = \displaystyle \frac{1}{2m_1} \| p_1 \|^2 + \frac{1}{2m_2} \| p_2 \|^2 + V(d(x_1, x_2)),
\]
with $m_1, m_2 > 0$, $p_1 \in T_{x_1}^*M$, $p_2 \in T_{x_2}^*M$, $(x_1, x_2) \in Q$ and $V: \R^+ \to \R$ a smooth function.
As in the paper,  $||\cdot ||$ denotes the norm in the fibers of $T^*Q$ induced by the Riemannian metric. On the other hand, $d: Q \to \R$ is the restriction to $Q$ of the Riemannian distance in $M\times M$. 

%\textcolor{Red}{It is clear that the Hamiltonian system $(T^*Q, H)$ is invariant with respect to the cotangent lift to $T^*Q$ of the action of $G$ on $Q$. Thus, if $z: \R \to T^*Q$ is a solution of the Hamilton equations then, using Noether theorem (see, for instance, \cite{AbMa}), we have that there exists $\mu \in {\frak g}^*$ such that
%\[
%\langle z(t), \xi_Q(q(t)) \rangle = \langle \mu,  \xi \rangle, \; \; \forall t \in \R,
%\]
%$\xi_Q$ being the infinitesimal generator of the action of $G$ on $Q$ associated with $\xi \in {\frak g}$. }

Now, denote by $r$ the map defined by $r = d \circ \pi_Q$, with $\pi_Q: T^*Q \to Q$ the canonical projection. Then, since $z(\R)$ is contained in the $G_\mu$-orbit of $z(0)$ and $G$ is an isometry group, we have that
\[
(r\circ z)(t) = (r \circ z)(0) : = r_0 > 0, \; \; \mbox{ for every } t \in \R.
\] 
Furthermore, we may prove that the Riemannian distance between two particles along a solution sufficiently close to  the $G_\mu$-stable relative equilibrium $z$ is almost equal to $r_0$.
\begin{proposition}\label{constant-distance}
For every $\epsilon > 0$, there exists an open subset $W$ in $T^*Q$ such that $z(\R) \subseteq W$ and for every solution $z'$ of the Hamiltonian system $(T^*Q, H)$ with initial condition in $W$ we have that
\[
| r(z'(t)) - r_0 | < \epsilon, \; \; \forall t.
\]
\end{proposition}

\begin{proof}
Consider the open subset of $T^*Q$
\[
U = r^{-1}(r_0 - \epsilon, r_0 + \epsilon).
\]
Then, $U$ is $G$-invariant and $z(\R) \subseteq U$. As a consequence of $G_\mu$-stability of the relative equilibrium $z$, there exists an open subset $W \subseteq T^*Q$, with $z(\R) \subseteq W$, such that
\[
z'(t) \in U, \; \; \forall t,
\]
for any solution  $z'$  with
initial condition in $W$.
Therefore,
\[
| r(z'(t)) - r_0 | < \epsilon, \; \; \forall t.
\]
\end{proof}
%\textcolor{Blue}{ We remark that the above result also holds for $H$-stable relative equilibria where $H$ is any subgroup of $G$ since $H$-stability implies $G$-stability. Moreover}, 
It is easy to prove that a natural extension of the previous result also holds for the more general case of the $n$-body problem in $M$, with $ n \geq 2$.

\subsection*{Acknowledgments} We thank Miguel Teixid\'o-Rom\'an for a useful conversation and Florin Diacu for his suggestions to improve our paper. 
We also are thankful to R. Ch\'avez-Tovar for his help with the edition of the Figures in the paper. This work has been
partially supported by the project  PAPIIT IA103815, 
by MEC (Spain) Grants MTM2011-22585, MTM2012-34478, MTM2014-54855-P and the European Community IRSES-project GEOMECH-246981. EPC has received partial support by the Asociaci\'on Mexicana de Cultura A.C.
LGN and EPC acknowledge to the Departamento de Matem\'aticas, Estad{\'\i}stica e IO,
at Universidad de La Laguna, for its hospitality in numerous visits. JCM acknowledges to IIMAS
at Universidad Nacional Aut\'onoma de M\'exico and the Departamento de Matem\'aticas at Universidad Aut\'onoma Metropolitana-Iztapalapa de M\'exico, for its hospitality in a visit
where this paper was started. MRO acknowledges the financial support of the EU Reintegration Grant PERG-GA-2010-27697.

\end{document}